\documentclass{irmaart}
\input{ha.sty}
\usepackage{gastex}
\usepackage{lineno}
    \linenumbers
\usepackage[hypertex,hyperindex,pagebackref,final]{hyperref}
\usepackage{calc}
\newcounter{hours}\newcounter{minutes}
\newcommand\printtime{\setcounter{hours}{\time/60}%
  \setcounter{minutes}{\time-\value{hours}*60}%
  \thehours\,h\,\theminutes}
\newcommand\dateandtime{\today\quad\printtime}
\newcommand\A{\mathcal{A}}
\newcommand\B{\mathcal{B}}
\renewcommand\C{\mathcal{C}}

\newcommand\D{\mathcal{D}}
\newcommand\J{\mathcal{J}}
\newcommand\E{\mathcal{E}}
\newcommand\F{\mathcal{F}}
\newcommand\K{\mathcal{K}}
\renewcommand\R{\mathcal{R}}
\renewcommand\L{\mathcal{L}}
\renewcommand\H{\mathcal{H}}

\newcommand\Card{{\rm Card}}
\newcommand\Id{{\rm Id}}
\let\edge\xrightarrow
\newcommand\intertitre[1]{\paragraph{\textit{#1}\/.}}
\usepackage{gastexpatch}
\gasset{Nadjustdist=1.5,loopdiam=7,loopwidth=6}
\usepackage{ifthen}
\newboolean{colorprint}

\begin{document}

\title{Symbolic dynamics}

\author{Marie-Pierre B\'eal$^{1}$, Jean Berstel$^{1}$, S{\o}ren Eilers$^{2}$,
Dominique Perrin$^{1}$}

\markboth{M.-P. B\'eal, J. Berstel, S. Eilers, D. Perrin}{Symbolic dynamics}

\address{ $^1$LIGM (Laboratoire d'Informatique Gaspard-Monge),
  Universit\'e Paris-Est\\ 
  email:\,\url{{beal,berstel,perrin}@univ-mlv.fr}
  \\[4pt]
  $^2$Institut for Matematiske Fag, K{\o}benhavns Universitet\\
  email:\,\url{eilers@math.ku.dk}}

\maketitle\label{chapterBBEP} 
\begin{center}
  \dateandtime
\end{center}



\localtableofcontents
\vfill
%
%

\section{Introduction}
Symbolic dynamics is part of dynamical systems theory. It studies 
discrete dynamical systems called shift spaces and their relations
under appropriately defined morphisms, in particular isomorphisms
called conjugacies. A special emphasis has been
put on the classification of shift spaces up to conjugacy or flow equivalence.

There is a considerable overlap between symbolic dynamics and automata
theory. Actually, one of the basic objects of symbolic dynamics, the
sofic systems, are essentially the same as finite automata. In
addition, the morphisms of shift spaces are a particular case of
rational transductions, that is functions defined by finite automata
with output. The difference is that symbolic dynamics considers mostly
infinite words and that all states of the automata are initial and
final. Also, the morphisms are particular transductions which are
given by local maps.

This chapter presents some of the links between automata theory and
symbolic dynamics. The emphasis is on two particular points. The
first one is the interplay between some particular classes of
automata, such as local automata and results on embeddings
of shifts of finite type. The second one is the connection between
syntactic semigroups and the classification of sofic shifts
up to conjugacy.

The chapter is organized as follows. 
In Section~\ref{sectionShiftSpaces},
we introduce the basic notions of symbolic dynamics: shift spaces,
conjugacy and flow equivalence. We state without proof two important
results: the Decomposition Theorem and the Classification Theorem.

In Section~\ref{sectionAutomata}, we introduce automata in relation to
sofic shifts. In Section~\ref{sectionMinimalAutomata}, we define
two kinds of minimal automata for shift spaces: the Krieger automaton
and the Fischer automaton. We also relate these automata with the
syntactic semigroup of a shift space. 

In Section~\ref{sectionSymbolicConjugacy},
we state and prove an analogue due to Nasu of the Decomposition
Theorem and of the Classification Theorem.

In Section~\ref{sectionSpecialFamilies} we consider two special
families of automata: local automata and automata with finite delay.
We show that they are related to shifts of finite  type and of
almost finite type, respectively. We prove an embedding theorem
(Theorem~\ref{BBEP:TheoremLocal}) which is a counterpart for automata
of a result known as Nasu's masking lemma.

In Section~\ref{sectionSyntacticInvariants} we study syntactic
invariants of sofic shifts. We introduce the syntactic graph of an
automaton.  We show that that the syntactic graph of an automaton is
invariant under conjugacy (Theorem~\ref{theoremSyntGraphAutomata}) and
also under flow equivalence. We finally state some results concerning
the shift spaces corresponding to some pseudovarieties of ordered
semigroups.

We follow the notation of the book of Doug Lind and Brian
Marcus~\cite{Lind&Marcus:1995}.  In general, we have not not
reproduced the proofs of the results which can be found there.  We
thank Mike Boyle and Alfredo Costa for their help.

%
%
\section{Shift spaces}\label{sectionShiftSpaces}

This section contains basic definitions concerning symbolic dynamics.

The first subsection gives the definition of shift spaces, and the
important case of edge shifts. 

The next
subsection and thus also under (Section~\ref{subSectionConjugacy}) introduces conjugacy,
and the basic notion of state splitting and merging. It contains
the statement of two important theorems, the Decomposition Theorem
(Theorem~\ref{ShiftDecompositionTheorem}) and the Classification
Theorem (Theorem~\ref{ClassificationTheorem}).

The last subsection (Section~\ref{subSectionFlowEquivalence})
introduces flow equivalence, and states Frank's characterization of
flow equivalent edge shifts (Theorem~\ref{thmFranks}).

\subsection{Shift spaces}\label{subSectionShiftspace}

Let $A$ be a finite alphabet.  We  denote by $A^*$ the set of words on
$A$ and by $A^+$ the set of nonempty words. A word $v$ is a \emph{factor} of a word $t$ if $t=uvw$ for some
words $u,w$. 

We denote by $A^\Z$ the set of
biinfinite sequences of symbols from $A$.  This set is a topological
space in the product topology of the discrete topology on $A$.
The \emph{shift transformation}\index{shift space!transformation} on $A^\Z$ is the map
$\sigma_A$ from $A^\Z$ onto itself defined by $y=\sigma_A(x)$ if
$y_n=x_{n+1}$ for $n\in \Z$.
A set $X\subset A^\Z$ is \emph{shift invariant} if $\sigma(X)=X$.
A \emph{shift space}\index{shift space} on the alphabet $A$ is a shift-invariant subset of
$A^\Z$ which is closed in the topology.
The set $A^\Z$ itself is a shift space called the \emph{full
  shift}\index{full shift}\index{shift space!full}.

For a set $W\subset A^*$ of words (whose elements are called the
\emph{forbidden factors}\index{forbidden factors}\index{shift
  space!forbidden factors}), we denote by $X^{(W)}$ the set of $x\in
A^\Z$ such that no $w\in W$ is a factor of $x$.

\begin{proposition}
  The shift spaces on the alphabet $A$ are the sets $X^{(W)}$, for $W\subset A^*$.
\end{proposition}
A shift space $X$ is of \emph{finite
  type}\index{shift space!finite type}  if there is a
finite set  $W\subset A^*$ such that $X=X^{(W)}$.
\begin{example}\label{ExGoldenMeanShift}
  Let $A=\{a,b\}$, and let $W=\{bb\}$. The shift $X^{(W)}$ is composed
  of the sequences without two consecutive $b$'s. It is a shift of
  finite type, called the \emph{golden mean shift}\index{golden mean
    shift}\index{shift space!golden mean}.
\end{example}

Recall that a set $W\subset A^*$ is said to be \emph{recognizable} if
it can be recognized by a finite automaton or, equivalently, defined
by
a regular expression.
A shift space $X$ is said to be 
\emph{sofic}\index{shift space!sofic}\index{sofic shift} if there is
a recognizable set $W$ such that $X=X^{(W)}$.
Since a finite set is recognizable, any  shift of finite type is sofic.

\begin{example}\label{ExEvenShift}
  Let $A=\{a,b\}$, and let $W=a(bb)^*ba$. The shift $X^{(W)}$ is composed
  of the sequences where two consecutive occurrences of the symbol $a$
  are separated by an even number of $b$'s. It is a sofic shift called
  the \emph{even shift}\index{even shift}\index{shift space!even}. It is not
a shift of finite type. Indeed, assume that $X=X^{(V)}$ for a finite
set
$V\subset A^*$. Let $n$ be the maximal length of the words of $V$.
A biinfinite repetition of the word $ab^n$ has the same blocks
of length at most $n$ as a biinfinite repetition of the word
$ab^{n+1}$.
However, one is in $X$ if and only if the other is not in $X$, a contradiction.
\end{example}

\begin{example}
Let $A=\{a,b\}$ and let $W=\{ba^nb^ma\mid n,m\ge 1, n\ne m\}$. The shift $X^{(W)}$
is composed of infinite sequences of the form 
$\ldots a^{n_i}b^{n_i}a^{n_{i+1}}b^{n_{i+1}}\ldots$. The set $W$ is
not
recognizable and it can be shown that $X$ is not sofic.
\end{example}

\intertitre{Edge shifts} 

In this chapter, a \emph{graph}\index{graph} $G=(Q,\E)$ is a pair
composed of a finite set $Q$ of
\emph{vertices}\index{vertex}\index{graph!vertex} (or
\emph{states}\index{state}\index{graph!state}), and a finite set $\E$ of
\emph{edges}\index{edge}\index{graph!edge}. The graph is equipped with
two maps $i,t:\E\to Q$ which associate, to an edge $e$, its
\emph{initial} and \emph{terminal}
vertex\index{vertex!initial}\index{vertex!terminal}%
\index{edge!initial vertex}\index{edge!terminal vertex}\footnote{We
  avoid the use of the terms `initial state' or `terminal state' of an
  edge to avoid confusion with the initial or terminal states of an
  automaton}. We say that $e$ starts in $i(e)$ and ends in
$t(e)$. Sometimes, $i(e)$ is called the \emph{source}\index{source of
    an edge}\index{edge!source} and $t(e)$ is called the
  \emph{target}\index{target of an edge}\index{edge!target} of $e$.

  We also say that $e$ is an incoming edge for $t(e)$, and an outgoing
  edge for $i(e)$.  Two edges $e,e'\in\E$ are
  \emph{consecutive}\index{edge!consecutive} if $t(e)=i(e')$.

For $p,q\in Q$, we denote by $\E_p^{q}$ the set of edges of a graph
$G=(Q,\E)$ starting in state $p$ and ending in state $q$.  The
\emph{adjacency matrix}\index{adjacency matrix}%
\index{graph!adjacency matrix}\index{matrix!adjacency} of a graph
$G=(Q,\E)$ is the $Q\times Q$-matrix $M(G)$ with elements in $\N$
defined by
\begin{displaymath}
  M(G)_{pq}=\Card(\E_p^{q})\,.
\end{displaymath}

A (finite or biinfinite) \emph{path}\index{path in a
  graph}\index{graph!path} is a (finite or biinfinite) sequence of
consecutive edges. The \emph{edge shift}\index{edge shift}%
\index{shift space!edge} on the graph $G$ is the set of biinfinite
paths in $G$. It is denoted by $X_G$ and is a shift of finite type on
the alphabet of edges. Indeed, it can be defined by taking the set of
non-consecutive edges for the set of forbidden factors. The converse
does not hold, since the golden mean shift is not an edge shift.
However, we shall see below (Proposition~\ref{BBEP:prop:conj}) that
every shift of finite type is conjugate to an edge shift.

A graph is \emph{essential}%
\index{graph!essential}\index{essential!graph} if every state has at
least one incoming and one outgoing edge. This implies that every edge
is on a biinfinite path. The \emph{essential part} of a graph $G$
is the subgraph obtained by restricting to the set of vertices and
edges which are on a biinfinite path.

\subsection{Conjugacy}\label{subSectionConjugacy}

\intertitre{Morphisms}

Let $X$ be a shift space on an alphabet $A$, and let $Y$ be a shift
space on an alphabet $B$.

A \emph{morphism}\index{morphism!of shifts}\index{shift space!morphism}
$\varphi$ from $X$ into $Y$ is a continuous map from $X$ into $Y$
which commutes with the shift.  This means that
$\varphi\circ\sigma_A=\sigma_B\circ\varphi$.

Let $k$ be a positive integer. A
\emph{$k$-block}\index{block}\index{k-block@$k$-block} of $X$ is a
factor of length~$k$ of an element of $X$.  We denote by $\B(X)$ the
set of all blocks of $X$ and by $\B_{k}(X)$ the set of $k$-blocks of
$X$.  A function $f:\B_{k}(X)\to B$ is called a $k$-\emph{block
  substitution}\index{block!substitution}\index{substitution!block}.
Let now $m,n$ be fixed nonnegative integers with $k=m+1+n$. Then the
function $f$ defines a map $\varphi$ called \emph{sliding block
  map}\index{block map!sliding}\index{sliding block map}%
\index{map!sliding block} with \emph{memory}%
\index{block map!memory}\index{memory, block map} $m$ and
\emph{anticipation}\index{anticipation, block map}%
\index{block map!anticipation} $n$ as follows. The image of $x\in X$
is the element $y=\varphi(x)\in B^\Z$ given by
\begin{displaymath}
  y_i=f(x_{i-m}\cdots x_i\cdots x_{i+n})\,.
\end{displaymath}
We denote $\varphi=f_\infty^{[m,n]}$.  It is a sliding block map
from $X$ into $Y$ if $y$ is in $Y$ for all $x$ in $X$. We also say
that $\varphi$ is a $k$-block map from $X$ into $Y$. The
simplest case occurs when $m=n=0$. In this case, $\varphi$ is
a $1$-block map.

The following result is Theorem~6.2.9 in~\cite{Lind&Marcus:1995}.

\begin{theorem}[Curtis--Lyndon--Hedlund]%
  \index{Curtis--Lyndon--Hedlund Theorem}%
  \index{theorem!Curtis--Lyndon--Hedlund} A map from a shift space $X$
  into a shift space $Y$ is a morphism if and only if it is a sliding
  block map.
\end{theorem}

\intertitre{Conjugacies of shifts}

A morphism from a shift $X$ onto a shift $Y$ is called a
\emph{conjugacy}\index{conjugacy!of shifts}\index{shift space!conjugacy} if
it is one-to-one from $X$ onto $Y$. Note that in this case, using
standard topological arguments, one shows that the inverse mapping is
also a morphism, and thus a conjugacy.

We define the $n$-th \emph{higher block
  shift}\index{higher!block shift}\index{shift space!higher block} $X^{[n]}$
of a shift $X$ over the alphabet $A$ as follows. The alphabet of
$X^{[n]}$ is the set $B=\B_n(X)$ of blocks of length $n$ of $X$. 
\begin{proposition}
The shifts $X$ and $X^{[n]}$ for $n\ge 1$ are conjugate.
\end{proposition}
\begin{proof}
Let
$f:\B_n(X)\to B$ be the $n$-block substitution which maps the factor
$x_1\cdots x_n$ to itself, viewed as a symbol of the alphabet $B$. By
construction, the shift $X^{[n]}$ is the image of $X$ by the map
$f^{[n-1,0]}_\infty$. This map is a conjugacy since it is bijective, and
its inverse is the $1$-block map $g_\infty$ corresponding to the
$1$-block map which associates to the symbol $x_1\cdots x_n$ of $B$
the symbol $x_n$ of $A$. 
\end{proof}

Let $G=(Q,\E)$ be a graph. For an integer $n\ge 1$, denote by
$G^{[n]}$ the following graph called the $n$-th \emph{higher edge
  graph}\index{higher!edge graph}\index{graph!higher edge} of
$G$. For $n=1$, one has $G^{[1]}=G$. For $n>1$, the set of states of
$G^{[n]}$ is the set of paths of length $n-1$ in $G$. The edges of
$G^{[n]}$ are the paths of length~$n$ of $G$. The start state of an
edge $(e_1,e_2,\ldots,e_n)$ is $(e_1,e_2,\ldots,e_{n-1})$ and its end
state is $(e_2,e_3,\ldots,e_n)$.

The following result shows that the higher block shifts of an edge
shift are again edge shifts.
\begin{proposition}
Let $G$ be a graph. For $n\ge 1$, one has $X_G^{[n]}=X_{G^{[n]}}$. 
\end{proposition}
A shift of finite type need not be an edge shift. For example
the golden mean shift of Example~\ref{ExGoldenMeanShift} is not an edge shift.
However, any shift of finite type comes from an edge shift in the
following sense.
\begin{proposition}\label{BBEP:prop:conj}
  Every shift of finite type is conjugate to an edge shift.
\end{proposition}

\begin{proof}
  We show that for every shift of finite type $X$ there is an integer
  $n$ such that $X^{[n]}$ is an edge shift. Let $W\subset A^*$ be a
  finite set of words such that $X=X^{(W)}$, and let $n$ be the
  maximal length of the words of $W$. If $n=0$, $X$ is the full shift.
  Thus we assume $n\ge1$. Define a graph $G$ whose vertices are the
  blocks of length $n-1$ of $X$, and whose edges are the block of
  length $n$ of $X$. For $w\in\B_n(X)$, the initial (resp. terminal)
  vertex of $w$ is the prefix (resp. suffix) of length $n-1$ of $w$.
  
  We show that $X_G=X^{[n]}$. An element of $X^{[n]}$ is always an
  infinite path in $G$.  To show the other inclusion, consider
  an infinite path $y$ in $G$. It is the sequence of $n$-blocks of an
  element $x$ of $A^\Z$ which does not contain any block on $W$.
  Since $X=X^{(W)}$, we get that $x$ is in $X$. Consequently, $y$ is
  in $X^{[n]}$. This proves the equality.
\end{proof}

\begin{proposition}
A shift space that is conjugate to a shift of finite type is itself
of finite type.
\end{proposition}
\begin{proof}
Let $\varphi:X\rightarrow Y$ be a conjugacy from a shift of finite
type $X$ onto a shift space $Y$. By Proposition~\ref{BBEP:prop:conj}, we
may assume that $X=X_G$ for some graph $G$. Changing $G$ into
some higher edge graph, we may assume that $\varphi$ is $1$-block.
We may consider $G$ as a graph labeled by $\varphi$.
Suppose that $\varphi^{-1}$ has memory $m$ and anticipation $n$.
Set $\varphi^{-1}=f^{[m,n]}_\infty$.
Let $W$ be the set of words of length $m+n+2$ which are
not the label of a path in $G$. We show that $Y=X^{(W)}$, which
implies
that $Y$ is of finite type.  Indeed, the
inclusion
$Y\subset X^{(W)}$ is clear. Conversely, consider $y$  in $X^{(W)}$.
For each $i\in\Z$, set $x_i=f(y_{i-m}\cdots y_i\cdots y_{i+n})$.
Since $y_{i-m}\cdots y_i\cdots y_{i+n}y_{i+n+1}$ is the label of a
path in $G$, the edges $x_i$ and $x_{i+1}$ are consecutive. Thus
$x=(x_i)_{i\in \Z}$ is in $X$ and $y=\varphi(x)$ is in $Y$.
\end{proof}

\intertitre{Conjugacy invariants}

No effective characterization of conjugate shift spaces is known, even
for shifts of finite type. There are however several quantities that
are known to be invariant under conjugacy.

The \emph{entropy}\index{entropy}\index{shift space!entropy} of a
shift space $X$ is defined by
\begin{displaymath}
h(X)=\lim_{n\rightarrow\infty}\frac{1}{n}\log s_n\,,
\end{displaymath}
where $s_n=\Card(\B_n(X))$.  The limit exists because the sequence
$s_n$ is sub-additive (see \cite{Lind&Marcus:1995} Lemma 4.1.7). Note
that since $\Card(\B_n(X))\le\Card(A)^n$, we have
$h(X)\le\log\Card(A)$. If $X$ is nonempty, then $0\le h(X)$.

The following statement shows that the entropy is invariant under
conjugacy (see~\cite{Lind&Marcus:1995} Corollary 4.1.10).

\begin{theorem}
  If $X,Y$ are conjugate shift spaces, then $h(X)=h(Y)$.
\end{theorem}

\begin{example}
  Let $X$ be the golden mean shift of
  Example~\ref{ExGoldenMeanShift}. Then a block of length $n+1$ is
  either a block of length $n-1$ followed by $ab$ or a block of length
  $n$ followed by $a$. Thus $s_{n+1}=s_n+s_{n-1}$. As a classical
  result, $h(X)=\log\lambda$ where $\lambda=(1+\sqrt{5})/2$ is the
  golden mean.
\end{example}

An element $x$ of a shift space $X$ over the alphabet $A$ has
\emph{period}\index{period} $n$ if $\sigma_A^n(x)=x$. If $\varphi:X\to
Y$ is a conjugacy, then an element $x$ of $X$ has period $n$ if and
only if $\varphi(x)$ has period $n$.

The \emph{zeta function}%
\index{zeta function}\index{shift space!zeta function} of a shift space $X$
is the power series
\begin{displaymath}
\zeta_X(z)=\exp\sum_{n\ge 0}\frac{p_n}{n}z^n\,,
\end{displaymath}
where $p_n$ is the number of elements $x$ of $X$ of period $n$.

It follows from the definition that the sequence $(p_n)_{n\in\N}$ is
invariant under conjugacy, and thus the zeta function of a shift
space is invariant under conjugacy.

Several other conjugacy invariants are known. One of them is the
Bowen-Franks group of a matrix which defines an invariant of the
associated shift space. This will be defined below.

\begin{example}
  Let $X=A^\Z$. Then $\zeta_X(z)=\frac{1}{1-kz}$, where
  $k=\Card(A)$. Indeed, one has $p_n=k^n$, since an element $x$ of
  $A^\Z$ has period $n$ if and only if it is a biinfinite repetition
  of a word of length $n$ over $A$.
\end{example}


\intertitre{State splitting}

Let $G=(Q,\E)$ and $H=(R,\F)$ be graphs. A pair $(h,k)$ of surjective
maps $k:R\to Q$ and $h:\F\to\E$ is called a \emph{graph morphism}
\index{graph!morphism}\index{morphism!graph} from $H$ onto $G$ if the
two diagrams in Figure~\ref{f} are commutative.

\begin{figure}[hbt]
\centering
\gasset{Nframe=n}
\begin{picture}(15,15)
\node(f)(0,15){$\F$}\node(e)(15,15){$\E$}
\node(r)(0,0){$R$}\node(p)(15,0){$Q$}
\drawedge(f,e){$h$}\drawedge(f,r){$i$}\drawedge(e,p){$i$}\drawedge(r,p){$k$}
\end{picture}\qquad\qquad\qquad
\begin{picture}(15,15)
\node(f)(0,15){$\F$}\node(e)(15,15){$\E$}
\node(r)(0,0){$R$}\node(p)(15,0){$Q$}
\drawedge(f,e){$h$}\drawedge(f,r){$t$}\drawedge(e,p){$t$}\drawedge(r,p){$k$}
\end{picture}
\caption{Graph morphism.}\label{f}
\end{figure}
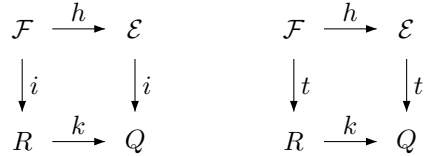

A graph morphism $(h,k)$ from $H$ onto $G$ is an
\emph{in-merge}\index{in-merge!graph morphism} from $H$ onto $G$ if
for each $p,q\in Q$ there is a partition $(\E_p^q(t))_{t\in
  k^{-1}(q)}$ of the set $\E_p^q$ with the following property.
 For each $r,t\in R$ 
and $p,q\in Q$ with $k(r)=p$, $k(t)=q$, the restriction of
the map $h$ to $\F_r^t$ is a bijection onto
$\E_p^q(t)$.  If this holds, then $G$ is called an
\emph{in-merge}\index{in-merge!of graph}\index{graph!in-merge} of $H$,
and $H$ is an \emph{in-split}%
\index{in-split!of graph}\index{graph!in-split} of $G$.\footnote{In
  this chapter, a \emph{partition}\index{partition} of a set $X$ is a
  family $(X_i)_{i\in I}$ of pairwise disjoint, possibly empty subsets
  of $X$, indexed by a set $I$, such that $X$ is the union of the sets
  $X_i$ for $i\in I$.}

Thus an in-split $H$ is obtained from a graph $G$ as follows: each
state $q\in Q$ is split into copies which are the states of $H$ in
the set $k^{-1}(q)$.  Each of these states $t$ receives a copy of
$\E_p^q(t)$ starting in $r$ and ending in $t$ for each $r$ in
$k^{-1}(p)$.

Each $r$ in $k^{-1}(p)$ has the same number of edges going out of $r$
and coming in $s$, for any $s\in R$. 
 
Moreover, for any $p,q\in Q$ and $e\in\E_p^q$, all edges in
$h^{-1}(e)$
have the same terminal vertex, namely the state $t$ such that $e\in\E_p^q(t)$.

\begin{example}\label{BBEP:exInMerge}
Let $G$ and $H$ be the graphs represented on
Figure~\ref{BBEP:figureInSplitGraphs}. Here $Q=\{1,2\}$ and
$R=\{3,4,5\}$. 
\ifthenelse{\boolean{colorprint}}{
\begin{figure}[hbt]
\centering
\gasset{Nadjust=wh}
\begin{picture}(75,30)(-10,-5)
  \put(0,-5){
    \begin{picture}(20,20)
      \node(1)(0,10){$1$}\node[linecolor=green](2)(20,10){$2$}      
      \drawloop[linecolor=blue](1){}
      \drawloop[loopangle=180,linecolor=blue](1){}
      \drawedge[curvedepth=3,linecolor=green](1,2){}
      \drawedge[curvedepth=2,linecolor=blue](2,1){}      
      \drawedge[curvedepth=5,linecolor=red](2,1){}
    \end{picture}
  }
  \put(40,0){
    \begin{picture}(20,25)
      \node[linecolor=blue](3)(0,15){$3$}
      \node[linecolor=green](5)(20,7){$5$}
      \node[linecolor=red](4)(0,0){$4$}       
      \drawloop[linecolor=blue](3){}
      \drawloop[loopangle=180,linecolor=blue](3){}
      \drawedge[curvedepth=2,linecolor=green](3,5){}
      \drawedge[curvedepth=2,linecolor=green](4,5){}
      \drawedge[curvedepth=2,linecolor=blue](5,3){}
      \drawedge[curvedepth=2,linecolor=blue](4,3){}
      \drawedge[curvedepth=-2,linecolor=blue](4,3){}      
      \drawedge[curvedepth=2,linecolor=red](5,4){}
    \end{picture}
  }
\end{picture}
\caption{An in-split from $G$ (on the left) onto $H$ (on the
  right).}\label{BBEP:figureInSplitGraphs}
\end{figure}
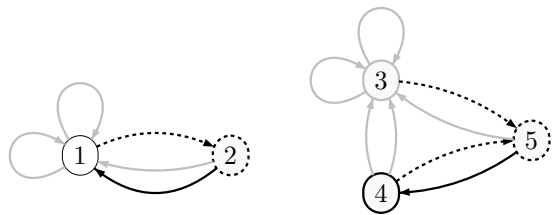
}{%
\begin{figure}[hbt]
\centering
\gasset{Nadjust=wh, linewidth=0.3}
\begin{picture}(75,30)(-10,-5)
  \put(0,-5){ 
    \begin{picture}(20,20)\gasset{fillcolor=lightgray!10}
      \node[linewidth=0.1,Nfill=n](1)(0,10){$1$}
      \node[dash={0.5 0.5}0](2)(20,10){$2$}      
      \drawloop[linecolor=lightgray](1){}
      \drawloop[loopangle=180,linecolor=lightgray](1){}
      \drawedge[curvedepth=3,dash={0.5 0.5}0](1,2){}
      \drawedge[curvedepth=2,linecolor=lightgray](2,1){}      
      \drawedge[curvedepth=5](2,1){}
    \end{picture}
  }
  \put(40,0){
    \begin{picture}(20,25)\gasset{fillcolor=lightgray!10}
      \node[linecolor=lightgray](3)(0,15){$3$}
      \node[dash={0.5 0.5}0](5)(20,7){$5$}
      \node(4)(0,0){$4$}        
      \drawloop[linecolor=lightgray](3){}
      \drawloop[loopangle=180,linecolor=lightgray](3){}
      \drawedge[curvedepth=2,dash={0.5 0.5}0](3,5){}
      \drawedge[curvedepth=2,dash={0.5 0.5}0](4,5){}
      \drawedge[curvedepth=2,linecolor=lightgray](5,3){}
      \drawedge[curvedepth=2,linecolor=lightgray](4,3){}
      \drawedge[curvedepth=-2,linecolor=lightgray](4,3){}      
      \drawedge[curvedepth=2](5,4){}
    \end{picture}
  }
\end{picture}
\caption{An in-split from $G$ (on the left) onto $H$ (on the
  right).}\label{BBEP:figureInSplitGraphs}
\end{figure}
}
The graph $H$ is an in-split of the graph $G$. The graph morphism
$(h,k)$ is defined by $k(3)=k(4)=1$ and $k(5)=2$. Thus the state $1$
of $G$ is split into two states $3$ and $4$ of $H$, and the map $h$ is
associated to the partition obtained as follows: the edges from $2$ to
$1$ are partitioned into two classes, indexed by $3$ and $4$
respectively, and containing each one edge from $2$ to $1$.
In the picture, the partitions are indicated by colors. The color of
an edge
on the right side corresponds to its terminal vertex. The color
of an edge on the left side is inherited through the graph morphism.
\end{example}
The following result is well-known (see~\cite{Lind&Marcus:1995}).
It shows that if $H$ is an in-split of a graph $G$, then $X_G$ and $X_H$ are
conjugate.

\begin{proposition}[\protect{\cite[Theorem 2.4.1]{Lind&Marcus:1995}}]\label{inMergingMap}
  If $(h,k)$ is an in-merge of a graph $H$ onto a graph $G$, then
  $h_\infty$ is a $1$-block conjugacy from $X_H$ onto $X_G$ and its
  inverse is $2$-block.
\end{proposition}
The map $h_\infty$ from $X_H$ to $X_G$ is called an \emph{edge
  in-merging map}\index{edge in-merging map}\index{map!edge
  in-merging} and its inverse an \emph{edge in-splitting
  map}\index{edge in-splitting map}\index{map!edge in-splitting}.

 A \emph{column division matrix}\index{column division
  matrix}\index{division matrix}\index{matrix!column division}
over two sets $R,Q$ is an $R\times Q$-matrix $D$ with elements in $\{0,1\}$
such that each column has at least one $1$ and each row has exactly
one $1$. Thus, the columns of such a matrix represent a partition of
$R$ into $\Card(Q)$ sets.

The following result is Theorem 2.4.14 of~\cite{Lind&Marcus:1995}.
\begin{proposition}\label{defEquivInMerge}
  Let $G$ and $H$ be essential graphs.  The graph $H$ is an in-split
  of the graph $G$ if and only if there is an $R\times Q$-column
  division matrix $D$ and a $Q\times R$-matrix $E$ with nonnegative
  integer entries such that
\begin{equation}\label{EqSplitGraph}
M(G)=ED,\quad M(H)=DE.
\end{equation}
\end{proposition}
\begin{example}
For the graphs $G,H$ of Example~\ref{BBEP:exInMerge}, one has
  $M(G)=DE$ and $M(H)=ED$
with
\begin{displaymath}
E=\begin{bmatrix}2&0&1\\1&1&0\end{bmatrix},\quad D=\begin{bmatrix}1&0\\1&0\\0&1\end{bmatrix}.
\end{displaymath}
\end{example}

Observe that a particular case of a column division matrix is a
permutation matrix. The corresponding in-split (or merge) is a
renaming of the states of a graph.

The notion of an \emph{out-merge} is defined symmetrically. A graph
morphism $(h,k)$ from $H$ onto $G$ is an
\emph{out-merge}\index{out-merge} from $H$ onto $G$ if for each
$p,q\in Q$ there is a partition $(\E_p^q(r))_{r\in k^{-1}(p)}$ of the
set $\E_p^q$ with the following property. For each $r,t\in R$,
and $p,q\in Q$ with $k(r)=p$, $k(t)=q$, the restriction of
the map $h$ to the set $\F_r^t$ is a bijection onto $\E_p^q(r)$.  If this
holds, then $G$ is called an \emph{out-merge} of $H$, and $H$ is an
\emph{out-split}\index{out-merge} of $G$.

Proposition~\ref{inMergingMap} also has a symmetrical version. Thus if
$(h,k)$ is an out-merge from $G$ onto $H$, then $h_\infty$ is a
$1$-block conjugacy from $X_H$ onto $X_G$ whose inverse is
$2$-block. The conjugacy $h_\infty$ is called an \emph{edge out-merging map}\index{map!edge out-merging} and its
inverse an \emph{edge out-splitting map}\index{map!edge out-splitting}.

Symmetrically, a \emph{row division matrix}\index{row division
  matrix}\index{division matrix}\index{matrix!row division} is a
matrix with elements in the set $\{0,1\}$ such that each column has at
least one $1$ and each row has exactly one $1$.

The following statement is symmetrical to
Proposition~\ref{defEquivInMerge}. 

\begin{proposition}
Let $G$ and $H$ be essential graphs.
The graph $H$ is an out-split of the graph $G$ if and only 
 if there is a row division matrix $D$ and a matrix $E$
with nonnegative integer entries
such that 
\begin{equation}\label{outSplitGraph}
M(G)=DE,\quad M(H)=ED.
\end{equation}
\end{proposition}

\begin{example}
Let $G$ and $H$ be the graphs represented on
Figure~\ref{BBEP:figureOutSplitGraphs}. Here $Q=\{1,2\}$ and $R=\{3,4,5\}$.
\ifthenelse{\boolean{colorprint}}{%
  \begin{figure}[hbt]
    \centering
    \gasset{Nadjust=wh}
    \begin{picture}(75,30)(-10,-5)
      \put(0,0){
        \begin{picture}(20,20)
          \node(1)(0,10){$1$}\node(2)(20,10){$2$}
          \gasset{linecolor=blue}
          \drawloop[loopangle=-90](1){}\drawloop[loopangle=90](1){}
          \gasset{linecolor=green}
          \drawedge[curvedepth=2](2,1){}
          \gasset{AHangle=30,linewidth=.5,linecolor=red}
          \drawloop[loopangle=180](1){}
          \drawedge[curvedepth=3](1,2){}
        \end{picture}
      }
      \put(40,5){
        \begin{picture}(20,25)
          \node[linecolor=red](3)(0,15){$3$}\node[linecolor=green](5)(20,7){$5$}
          \node[linecolor=blue](4)(0,0){$4$}
          \gasset{linecolor=blue}
          \drawloop[loopangle=180](4){}\drawloop[loopangle=-90](4){}
          \drawedge[curvedepth=2](4,3){}
          \drawedge[curvedepth=-2](4,3){}
          \gasset{linecolor=green}
          \drawedge[curvedepth=2](5,4){}
          \drawedge[curvedepth=2](5,3){}
          \gasset{AHangle=30,linecolor=red}
          \drawloop[loopangle=180](3){}
          \drawedge[curvedepth=-5](3,4){}
          \drawedge[curvedepth=2](3,5){}
        \end{picture}
      }
    \end{picture}
    \caption{The graphs $G$ and $H$.}\label{BBEP:figureOutSplitGraphs}
  \end{figure}
}{%
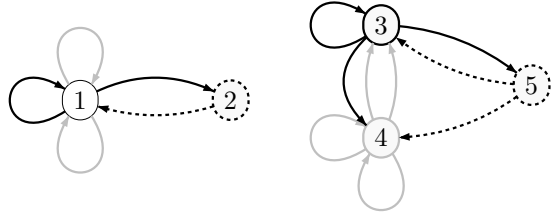
\begin{figure}[hbt]
  \centering
  \gasset{Nadjust=wh, linewidth=0.3}
  \begin{picture}(75,30)(-10,-5)
    \gasset{fillcolor=lightgray!10}
    \put(0,0){
      \begin{picture}(20,20)
        \node[linewidth=0.1,Nfill=n](1)(0,10){$1$}
        \node[dash={0.5 0.5}0](2)(20,10){$2$}
        \drawloop[loopangle=-90,linecolor=lightgray](1){}
        \drawloop[loopangle=90,linecolor=lightgray](1){}
        \drawedge[dash={0.5 0.5}0,curvedepth=2](2,1){}
        \drawloop[loopangle=180](1){}
        \drawedge[curvedepth=3](1,2){}
      \end{picture}
    }
    \put(40,5){
      \begin{picture}(20,25)
        \node(3)(0,15){$3$}
        \node[dash={0.5 0.5}0](5)(20,7){$5$}
        \node[linecolor=lightgray](4)(0,0){$4$}
        \drawloop[linecolor=lightgray,loopangle=180](4){}
        \drawloop[linecolor=lightgray,loopangle=-90](4){}
        \drawedge[linecolor=lightgray,curvedepth=2](4,3){}
        \drawedge[linecolor=lightgray,curvedepth=-2](4,3){}
        \drawedge[dash={0.5 0.5}0,curvedepth=2](5,4){}
        \drawedge[dash={0.5 0.5}0,curvedepth=2](5,3){}       
        \drawloop[loopangle=180](3){}
        \drawedge[curvedepth=-5](3,4){}
        \drawedge[curvedepth=2](3,5){}
      \end{picture}
    }
  \end{picture}
  \caption{The graphs $G$ and $H$.}\label{BBEP:figureOutSplitGraphs}
\end{figure}
}
The graph $H$ is an out-split of the graph $G$. The graph morphism
$(h,k)$
is defined by $k(3)=k(4)=1$ and $k(5)=2$. The map $h$ is associated
with the partition indicated by the colors. The color of an edge on
the
right side corresponds to its initial vertex. On the left side, the
color is inherited through the graph morphism.
One has $M(G)=ED$ and $M(H)=DE$ with
\begin{displaymath}
D=\begin{bmatrix}1&1&0\\0&0&1\end{bmatrix},\quad E=\begin{bmatrix}1&1\\2&0\\1&0\end{bmatrix} .
\end{displaymath}
\end{example}

We use the term \emph{split}  to mean either an in-split or an out-split.
The same convention holds for a \emph{merge}.

\begin{proposition}\label{stExtensionShift}
  For $n\ge 2$, the graph $G^{[n-1]}$ is an in-merge of the
  graph $G^{[n]}$.
\end{proposition}
\begin{proof}
  Consider for $n\ge 2$ the equivalence on the states of $G^{[n]}$
  which relates two paths of length $n-1$ which differ
  only by the first edge. It is clear that this equivalence is such
  that two equivalent elements have the same output. Thus $G^{[n-1]}$
  is an in-merge of $G^{[n]}$.
\end{proof}

\intertitre{The Decomposition Theorem}

The following result is known as the \emph{Decomposition
  Theorem}\index{Decomposition Theorem}\index{theorem!Decomposition}
(Theorem 7.1.2 in \cite{Lind&Marcus:1995}).
\begin{theorem}\label{ShiftDecompositionTheorem} Every
  conjugacy from an edge shift onto another is the composition of a
  sequence of edge splitting maps followed by a sequence of edge merging maps.
\end{theorem}
The statement of Theorem~\ref{ShiftDecompositionTheorem} given in
\cite{Lind&Marcus:1995} is less precise, since it does not specify the
order of splitting and merging maps.

The proof relies on the following statement (Lemma 7.1.3 in~\cite{Lind&Marcus:1995}).
\begin{lemma}\label{lemmaDecomp}
Let $G,H$ be graphs and
let $\varphi:X_G\rightarrow X_H$ be a $1$-block conjugacy whose
inverse has memory $m\ge 1$ and anticipation $n\ge 0$. There are
in-splittings $\overline{G},\overline{H}$ of the graphs $G,H$ and a
$1$-block conjugacy with memory $m-1$ and anticipation $n$
$\overline{\varphi}:X_{\overline{G}}\rightarrow X_{\overline{H}}$ such that
the following diagram commutes.
\begin{figure}[hbt]
\centering
\gasset{Nframe=n}
\begin{picture}(20,20)
\node(G)(0,20){$X_G$}\node(tG)(20,20){$X_{\overline{G}}$}
\node(H)(0,0){$X_H$}\node(tH)(20,0){$X_{\overline{H}}$}
\drawedge(G,tG){}\drawedge(H,tH){}
\drawedge(G,H){$\varphi$}\drawedge(tG,tH){$\overline{\varphi}$}
\end{picture}
\end{figure}
\end{lemma}
The horizontal edges in the above diagram represent the edge in-splitting
maps from $X_G$ to $X_{\overline{G}}$ and from $X_H$ to $X_{\overline{H}}$
respectively.
\intertitre{The Classification Theorem}

Two nonnegative integral square matrices $M,N$ are \emph{elementary
  equivalent}%
\index{elementary equivalent matrices}%
\index{matrix!elementary equivalent} if there exists a pair $R,S$ of
nonnegative integral matrices such that
\begin{displaymath}
  M=RS\,,\quad N=SR\,.
\end{displaymath}
Thus if a graph $H$ is a split of a graph $G$, then, by
Proposition~\ref{defEquivInMerge}, the matrices $M(G)$ and $M(H)$ are
elementary equivalent.  The matrices $M$ and $N$ are \emph{strong
  shift equivalent}%
\index{strong shift equivalent matrices}%
\index{matrix!strong shift equivalent} if there is a sequence
$(M_0,M_1,\ldots,M_n)$ of nonnegative integral matrices such that
$M_i$ and $M_{i+1}$ are elementary equivalent for $0\le i<n$ with $M_0=M$
and $M_n=N$.

The following theorem is Williams' Classification Theorem (Theorem
7.2.7 in~\cite{Lind&Marcus:1995})\index{Classification
  Theorem}\index{theorem!Classification}.

\begin{theorem}\label{ClassificationTheorem}
  Let $G$ and $H$ be two graphs. The edge shifts $X_G$ and $X_H$ are
  conjugate if and only if the matrices $M(G)$ and $M(H)$ are strong
  shift equivalent.
\end{theorem}

Note that one direction of this theorem is contained in the
Decomposition Theorem. Indeed, if $X_G$ and $X_H$ are conjugate, there
is a sequence of edge splitting and edge merging maps from $X_G$ to $X_H$. And
if $G$ is a split or a merge of $H$, then $M(G)$ and $M(H)$ are
elementary equivalent, whence the result in one direction follows.
Note also that, in spite of the easy definition of strong shift
equivalence, it is not even known whether there exists a decision
procedure for determining when two nonnegative integral matrices are
strong shift equivalent.

\subsection{Flow equivalence}\label{subSectionFlowEquivalence}

In this section, we give basic definitions and properties concerning
flow equivalence of shift spaces. The notion comes from the notion of
equivalence of continuous flows, see Section~13.6 of
\cite{Lind&Marcus:1995}. A characterization of flow equivalence for
 shift spaces (which we will take below as our definition of flow
equivalence for shift spaces) 
 is due to Parry and Sullivan~\cite{ParrySullivan:1975}. It is
noticeable that the flow equivalence of irreducible shifts of finite
type has an effective characterization, by Franks' Theorem
(Theorem~\ref{thmFranks}).

Let $A$ be an alphabet and $a$ be a letter in $A$.  Let $\omega$ be a
letter which does not belong to $A$. Set $B=A\cup\omega$. The
\emph{symbol expansion}\index{symbol!expansion}
\index{expansion!symbol} of a set $W\subset A^+$ relative to $a$ is
the image of $W$ by the semigroup morphism $\varphi:A^+\rightarrow
B^+$ such that $\varphi(a)=a\omega$ and $\varphi(b)=b$ for all $b\in
A\setminus a$.  Recall that a \emph{semigroup morphism}%
\index{semigroup morphism}\index{morphism!semigroup} $f:A^+\to B^+$ is
a map satisfying $f(xy)=f(x)f(y)$ for all words $x,y$. It should not be
confused with the morphisms of shift spaces defined earlier.  The
semigroup morphism $\varphi$ is also called a symbol expansion.  Let
$X$ be a shift space on the alphabet $A$. The \emph{symbol expansion}
of $X$ relative to $a$ is the least shift space $X'$ on the alphabet
$B=A\cup\omega$ which contains the symbol expansion of $\B(X)$.  Note
that if $\varphi$ is a symbol expansion, it defines a bijection from
$\B(X)$ onto $\B(X')$. The inverse of a symbol expansion is called a
\emph{symbol
  contraction}\index{symbol!contraction}\index{contraction!symbol}.

Two shift spaces $X,Y$ are said to be \emph{flow
  equivalent}\index{flow equivalent}\index{shift space!flow equivalent} if there
is a sequence $X_0,\ldots,X_n$ of shift spaces such that $X_0=X$,
$Y_n=Y$
and for $0\le i\le n-1$, either $X_{i+1}$ is  the image of $X_i$
by a conjugacy, a symbol expansion or a symbol contraction.

\begin{example}
  Let $A=\{a,b\}$. The symbol expansion of the full shift $A^\Z$
  relative to $b$ is conjugate to the golden mean shift.  Thus the
  full shift on two symbols and the golden mean shift are flow
  equivalent.
\end{example}

For edge shifts, symbol expansion can be replaced by another
operation. Let $G$ be a graph and let $p$ be a vertex of $G$.  The
\emph{graph expansion}\index{graph!expansion}\index{expansion!graph}
of $G$ relative to $p$ is the graph $G'$ obtained by replacing $p$ by
an edge from  a new vertex $p'$ to $p$ to and replacing all  edges
coming in
 $p$ by  edges coming in $p'$ (see
Figure~\ref{figGraphExpansion}). The inverse of a graph expansion is
called a \emph{graph contraction}\index{graph!contraction}%
\index{contraction!graph}.
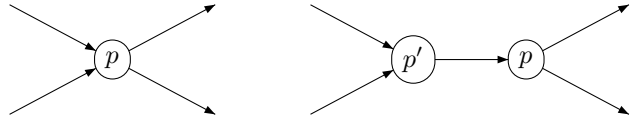
\begin{figure}[hbt]
\centering
\begin{picture}(80,20)
\gasset{Nframe=n,Nadjust=wh}
\node(1)(0,0){}\node(2)(0,16){}\node[Nframe=y](3)(15,8){$p$}\node(4)(30,0){}
\node(5)(30,16){}
\drawedge(1,3){}\drawedge(2,3){}\drawedge(3,4){}\drawedge(3,5){}

\node(6)(40,0){}\node(7)(40,16){}
\node[Nframe=y](8)(55,8){$p'$}\node[Nframe=y](9)(70,8){$p$}
\node(10)(85,0){}\node(11)(85,16){}
\drawedge(6,8){}\drawedge(7,8){}\drawedge(8,9){}\drawedge(9,10){}\drawedge(9,11){}
\end{picture}
\caption{Graph expansion}\label{figGraphExpansion}
\end{figure}
Note that graph expansion (relative to vertex 1) changes the adjacency matrix of a graph as
indicated below.
\begin{displaymath}
\begin{bmatrix}
a_{11}&a_{12}&\ldots&a_{1n}\\
a_{21}&a_{22}&\ldots&a_{2n}\\
\vdots&     &      &      \\
a_{n1}&a_{n2}&\ldots&a_{nn}
\end{bmatrix}
\longrightarrow
\begin{bmatrix}
0&a_{11}&a_{12}&\ldots&a_{1n}\\
1&0    &0     &\ldots&0\\
0&a_{21}&a_{22}&\ldots&a_{2n}\\
\vdots &     &&      &     \\
0&a_{n1}&a_{n2}&\ldots&a_{nn}
\end{bmatrix}
\end{displaymath}
\begin{proposition}
The flow equivalence relation on edge shifts is generated by
conjugacies
and graph expansions.
\end{proposition}
\begin{proof}
Let $G=(Q,E)$ be a graph and let $p$ be a vertex of $G$. The graph
expansion
of $G$ relative to $p$ can be obtained by a symbol expansion of
each of the edges coming into $p$ followed by a conjugacy which
merges all the new symbols into one new symbol.
Conversely, let $e$ be an edge of $G$. The symbol expansion of $X_G$
relative to $e$ can be obtained by a input split which makes $e$
the only edge going into its end vertex $q$ followed by a graph
expansion relative to $q$.
\end{proof}

The \emph{Bowen-Franks group}\index{Bowen-Franks group}
 of a square $n\times n$-matrix $M$ with
integer elements is 
the Abelian group 
\begin{displaymath}
BF(M)=\Z^n/\Z^n(I-M)
\end{displaymath}
where $\Z^n(I-M)$ is the image of $\Z^n$ under the matrix $I-M$ acting
on the right. In other terms, $\Z^n(I-M)$ is the Abelian group
generated by the rows of the matrix $I-M$. This notion is due to Bowen
and Franks~\cite{Bowen&Franks:1977}, who have shown that it is an
invariant for flow equivalence.

The following result is due to
Franks~\cite{Franks:1984}\index{theorem!Franks'}.
We say that a graph is \emph{trivial} if it is  reduced to one
cycle.
\begin{theorem}\label{thmFranks}
  Let $G,G'$ be two strongly connected nontrivial graphs and let $M,M'$ be their
  adjacency matrices.  The edge shifts $X_G,X_{G'}$ are flow
  equivalent if and only if $\det(I-M)=\det(I-M')$ and the groups $BF(M)$,
  $BF(M')$ are isomorphic.
\end{theorem}

In the case trivial graphs, the theorem is false. Indeed, any two
edge shifts on strongly connected
trivial graphs are flow equivalent and
are not flow equivalent to any edge shift on a nontrivial irreducible
graph.
 For any trivial
graph $G$ with adjacency matrix $M$, one has $\det(I-M)=0$ and
$BF(M)\sim \Z$. However there are nontrivial strongly connected graphs
such that $\det(I-M)=0$ and $BF(M)\sim \Z$.

The case of arbitrary shifts of finite type has been solved by 
 Huang (see~\cite{Boyle:2002,Boyle&Huang:2003}).  A similar characterization for
sofic shifts is not known (see~\cite{Boyle:2008}).

\begin{example}
Let 
\begin{displaymath}
M=\begin{bmatrix}4&1\\1&0\end{bmatrix},\quad 
M'=\begin{bmatrix}3&2\\1&0\end{bmatrix}.
\end{displaymath}
One has $\det(I-M)=\det(I-M')=-4$. Moreover $BF(M)\sim
\Z/4\Z$. Indeed, the rows of the matrix
$I-M$ are $\begin{bmatrix}-3&-1\end{bmatrix}$ and
  $\begin{bmatrix}-1&1\end{bmatrix}$. They generate the same group as
$\begin{bmatrix}4&0\end{bmatrix}$ and
      $\begin{bmatrix}-1&1\end{bmatrix}$.
Thus $BF(M)\sim\Z/4\Z$. In the same way, $BF(M')\sim\Z/4\Z$.
Thus, according to Theorem~\ref{thmFranks}, the edge shifts $X_G$
and $X_{G'}$ are flow equivalent.

Actually
$X_G$ and $X_{G'}$ are both flow equivalent to the full shift on $5$ symbols.
\end{example}
%
%
\section{Automata}\label{sectionAutomata}

In this section, we start with the definition and notation for
automata recognizing shifts, and we show that sofic shifts are
precisely the shifts recognized by finite automata
(Proposition~\ref{BBEP:PropSoficRec}). 

We introduce the notion of labeled conjugacy; it is a conjugacy
preserving the labeling. We extend the Decomposition Theorem and the
Classification Theorem to labeled conjugacies
(Theorems~\ref{AutomataDecompositiontheorem} and
\ref{AutomataClassificationtheorem}).  
\subsection{Automata and sofic shifts}

The automata considered in this section are finite automata. We do not
mention the initial and final states  in the notation when all
states are both initial and final. Thus, an automaton is denoted by
$\A=(Q,E)$ where $Q$ is the finite set of
\emph{states}\index{automaton!state} and $E\subset Q\times A\times Q$
is the set of \emph{edges}\index{automaton!edge}. The edge $(p,a,q)$
has initial state $p$, label $a$ and terminal state $q$.  The
underlying graph of $\A$ is the same as $\A$ except that the labels of
the edges are not used.

An automaton is \emph{essential}%
\index{automaton!essential}\index{essential!automaton} if 
its underlying graph is essential. The \emph{essential part}
of an automaton is its restriction to the essential part of
its underlying graph.

We denote by $X_\A$ the set of biinfinite paths in $\A$. It is the
edge shift of the underlying graph of $\A$.  Note that since the
automaton is supposed finite, the shift space $X_\A$ is on a finite
alphabet,
as required for a shift space. We denote by $L_\A$ the
set of labels of biinfinite paths in $\A$. We denote by $\lambda_\A$
the $1$-block map from $X_\A$ into the full shift $A^\Z$ which assigns
to a path its label. Thus $L_\A=\lambda_\A(X_\A)$.  If this holds, we
say that $L_\A$ is the shift space \emph{recognized} by
$\A$.\index{automaton!shift recognized}%
\index{shift space!recognized by an automaton}

The following
propositions describe how this notion of recognition is related to
that for finite words.  In the context of finite words, we denote by
$\A=(Q,I,E,T)$ an automaton with distinguished subsets $I$ (resp. $T$)
of initial (resp. terminal) states. A word $w$ is \emph{recognized} by $\A$
if there is a path from a state in $I$ to a state in $T$ labeled 
$w$. Recall that a set is recognizable if it is the set of words
recognized by a finite automaton.  An automaton $\A=(Q,I,T)$ is
\emph{trim}\index{trim automaton}\index{automaton!trim} if, for every
state $p$ in $Q$, there is a path from a state in $I$ to $p$ and a
path from $p$ to a state in $T$.

\begin{proposition}\label{BBEP:propAutomatonSofic1}
  Let $W\subset A^*$ be a recognizable set and let $\A=(Q,I,T)$ be a
  trim finite automaton recognizing the set $A^*\setminus A^*WA^*$.
  Then $L_\A=X^{(W)}$.
\end{proposition}

\begin{proof}
 The label of a biinfinite path in the
  automaton $\A$ does not contain a factor $w$ in $W$. Otherwise, there is a
  finite path $p\edge{w}q$ which is a segment of this infinite
  path. The path $p\edge{w}q$ can be extended to a path
  $i\edge{u}p\edge{w}q\edge{v}t$ for some $i\in I, t\in T$, and $uwv$ is
  accepted by $\A$, which is a contradiction.
  
  Next, consider a biinfinite word $x=(x_i)_{i\in\Z}$ in $X^{(W)}$. For
  every $n\ge0$, there is a path $\pi_n$ in the automaton $\A$ labeled
  $w_n=x_{-n}\cdots x_0\cdots{x_n}$ because the word $w_n$ has no
  factor in $W$. By compactness (K\"onig's lemma) there is an infinite
  path in $\A$ labeled $x$. Thus $x$ is in $L_\A$.
\end{proof}
The following proposition states in some sense the converse.

\begin{proposition}\label{BBEP:propAutomatonSofic0}
  Let $X$ be a sofic shift over $A$, and let $\A=(Q,I,T)$ be a trim
  finite automaton recognizing the set $\B(X)$ of blocks of $X$.  Then
  $L_\A=X$.
\end{proposition}

\begin{proof}
  Set $W=A^*\setminus \B(X)$. Then one easily checks that $X=X^{(W)}$.
  Next, $\A$ recognizes $A^*\setminus A^*WA^*$. By
  Proposition~\ref{BBEP:propAutomatonSofic1}, one has $L_\A=X$.
\end{proof}

\begin{proposition}\label{BBEP:PropSoficRec}
  A shift $X$ over $A$ is sofic if and only if there is a finite
  automaton $\A$ such that $X=L_\A$.
\end{proposition}

\begin{proof}
  The forward implication results from Proposition~\ref{BBEP:propAutomatonSofic1}.
  Conversely, assume that $X=L_\A$ for some finite automaton $\A$.
  Let $W$ be the set of finite words which are not labels of paths in
  $\A$. Clearly $X\subset X^{(W)}$. Conversely, if $x\in X^{(W)}$,
  then all its factors are labels of paths in $\A$. Again by
  compactness, $x$ itself is the label of a biinfinite path in~$\A$.
\end{proof}

\begin{example}\label{ExAutomataSofic}
  The golden mean shift of Example~\ref{ExGoldenMeanShift} is
  recognized by the automaton of Figure~\ref{figAutomatonEven} on the
  left while the even shift of Example~\ref{ExEvenShift} is recognized
  by the automaton of Figure~\ref{figAutomatonEven} on the right.
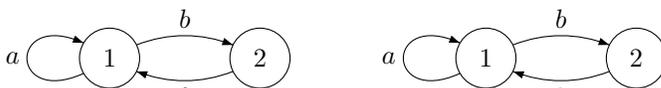
\begin{figure}[hbt]
\centering
\begin{picture}(70,10)
\put(0,0){
\begin{picture}(20,10)(0,-3)
\node(1)(0,0){$1$}\node(2)(20,0){$2$}
\drawloop[loopangle=180](1){$a$}\drawedge[curvedepth=3](1,2){$b$}
\drawedge[curvedepth=3](2,1){$a$}
\end{picture}
}
\put(50,0){
\begin{picture}(20,10)(0,-3)
\node(1)(0,0){$1$}\node(2)(20,0){$2$}
\drawloop[loopangle=180](1){$a$}\drawedge[curvedepth=3](1,2){$b$}
\drawedge[curvedepth=3](2,1){$b$}
\end{picture}
}
\end{picture}
\caption{Automata recognizing the golden mean and the even
  shift}\label{figAutomatonEven} 
\end{figure}
\end{example}

The \emph{adjacency matrix}\index{transition!
  matrix}\index{matrix!transition} of the automaton $\A=(Q,E)$ is the
$Q\times Q$-matrix $M(\A)$ with elements in $\N\langle A\rangle$
defined by
\begin{displaymath}
  (M(\A)_{pq},a)=
  \begin{cases}
    1&\text{if $(p,a,q)\in E$\,,}\\0&\text{otherwise.}
  \end{cases}
\end{displaymath}
We write $M$ for $M(\A)$ when the automaton is understood.
The entries in the matrix $M^n$, for $n\ge0$, have an easy
combinatorial interpretation: for each word $w$ of length $n$, the
coefficient $(M^n_{p,q},w)$ is the number of distinct paths from $p$
to $q$ carrying the label $w$.

A matrix $M$ is called
\emph{alphabetic}\index{alphabetic!matrix}\index{matrix!alphabetic}
over the alphabet $A$ if its elements are homogeneous polynomials of
degree~$1$ over $A$ with nonnegative coefficients. Adjacency matrices
are special cases of alphabetic matrices. Indeed, its elements are
homogeneous polynomials of degree~$1$ with coefficients $0$ or~$1$.

\subsection{Labeled conjugacy} 

Let $\A$ and $\B$ be two automata on the alphabet $A$.  A
\emph{labeled conjugacy}%
\index{labeled!conjugacy}\index{conjugacy!labeled} from $X_\A$ onto
$X_\B$ is a conjugacy $\varphi$ such that
$\lambda_\A=\lambda_\B\varphi$, that is such that the following diagram is
commutative.
\begin{figure}[hbt]
\centering
\gasset{Nframe=n}
\begin{picture}(20,25)
  \node(A)(0,20){$X_\A$}
  \node(B)(20,20){$X_\B$}
  \node(S)(10,0){$A^\Z$}
  \drawedge(A,B){$\varphi$}
  \drawedge[ELside=r](A,S){$\lambda_\A$}
  \drawedge(B,S){$\lambda_\B$}
\end{picture}
\end{figure}
We say that $\A$ and $\B$ are \emph{conjugate}%
\index{automaton!conjugate}%
\index{conjugate automata} if there exists a labeled conjugacy from
$X_\A$ to $X_\B$.  The aim of this paragraph is to give two
characterizations of labeled conjugacy.

\intertitre{Labeled split and merge}

Let $\A=(Q,E)$ and $\B=(R,F)$ be two automata. Let $G,H$ be the
underlying graphs of $\A$ and $\B$ respectively.

A \emph{labeled in-merge}\index{automaton!labeled in-merge}%
\index{labeled!in-merge}\index{in-merge!labeled} from $\B$ onto $\A$
is an in-merge $(h,k)$ from $H$ onto $G$ such that for each $f\in F$
the labels of $f$ and $h(f)$ are equal.  We say that $\B$ is a
\emph{labeled
  in-split}\index{labeled!in-split}\index{in-split!labeled}%
\index{automaton!labeled in-split} of $\A$, or that $\A$ is a
\emph{labeled in-merge} of $\B$.

The following statement is the analogue of
Proposition~\ref{inMergingMap} for automata.

\begin{proposition}\label{stLabeledInMergeIsCojugacy}
  If $(h,k)$ is a labeled in-merge from the automaton $\B$ onto the
  automaton $\A$, then the map $h_\infty$ is a labeled conjugacy from
  $X_\B$ onto $X_\A$.
\end{proposition}
\begin{proof}
Let $(h,k)$ be a labeled in-merge from $\B$ onto $\A$.
By Proposition~\ref{inMergingMap}, the map $h_\infty$ is
a $1$-block conjugacy from $X_\B$ onto $X_\A$. Since the labels
of $f$ and $h(f)$ are equal for each edge $f$ of $\B$, this map is a
labeled
conjugacy.
\end{proof}
The next statement is the analogue of
Proposition~\ref{defEquivInMerge} for automata.
\begin{proposition}\label{defEquivLabeledInSplit}
  An automaton $\B=(R,F)$ is a labeled in-split of the automaton
  $\A=(Q,E)$ if and only if there is an $R\times Q$-column division
  matrix $D$ and an alphabetic $Q\times R$-matrix $N$ such that
\begin{equation}\label{EqInSplitAutomata}
M(\A)=ND,\quad M(\B)=DN.
\end{equation}
\end{proposition}
\begin{proof}


  Suppose first that $D$ and $N$ are as described in the statement,
  and define a map $k:R\rightarrow Q$ by $k(r)=q$ if $D_{rq}=1$. We
  define $h:F\rightarrow E$ as follows. Consider an edge $(r,a,s)\in
  F$.  Set $p=k(r)$ and $q=k(s)$. Since $M(\B)=DN$, we have
  $(N_{ps},a)=1$.  Since $M(\A)=ND$, this implies that
  $(M(\A)_{pq},a)=1$ or, equivalently, that $(p,a,q)\in E$. We set
  $h(r,a,s)=(p,a,q)$. Then $(h,k)$ is a labeled in-merge. Indeed $h$
  is associated with the partitions defined by 
\begin{displaymath}
E_p^q(t)=\{(p,a,q)\in  E\mid (N_{pt},a)=1\mbox{ and } k(t)=q\}.
\end{displaymath}

Suppose conversely that $(h,k)$ is a labeled in-merge from $\B$ onto $\A$.
Let $D$ be the $R\times Q$-column division matrix defined by
\begin{displaymath}
D_{rq}=\begin{cases}1&\mbox{ if } k(r)=q\\0&\mbox{ otherwise
}\end{cases}
\end{displaymath}
For $p\in Q$ and $t\in R$, we define $N_{rt}$ as follows. Set $q=k(t)$.
By definition of an in-merge, there is a partition
$(E_p^q(t))_{t\in k^{-1}(q)}$ of $E_p^q$ such that $h$ is a bijection
from $F_r^t$ onto $E_p^q(t)$. For $a\in A$, set
\begin{displaymath}
(N_{pt},a)=\begin{cases}1&\mbox{ if } (p,a,q)\in E_p^q(t)\\0&\mbox{
    otherwise }\end{cases}
\end{displaymath}
Then $M(\A)=ND$ and $M(\B)=DN$.
\end{proof}

\ifthenelse{\boolean{colorprint}}{%
\begin{figure}[hbt]
  \centering
  \begin{picture}(75,35)(-12,-8)
  \gasset{Nadjust=wh}
    \put(0,-5){
      \begin{picture}(20,20)
        \node(1)(0,10){$1$}\node[linecolor=green](2)(20,10){$2$}
        \drawloop[linecolor=blue](1){$a$}
        \drawloop[loopangle=180,linecolor=blue](1){$c$}
        \drawedge[curvedepth=3,linecolor=green](1,2){$b$}
        \drawedge[curvedepth=3,linecolor=red](2,1){$a$}
      \end{picture}
    }
    \put(40,0){
      \begin{picture}(20,30)
        \node[linecolor=blue](3)(0,15){$3$}
        \node[linecolor=green](5)(20,5){$5$}
        \node[linecolor=red](4)(0,-5){$4$}
        \drawloop[linecolor=blue](3){$a$}
        \drawloop[loopangle=180,linecolor=blue](3){$c$}
        \drawedge[curvedepth=3,linecolor=green](3,5){$b$}
        \drawedge[curvedepth=3,linecolor=red](5,4){$a$}
        \drawedge[curvedepth=3,linecolor=green](4,5){$b$}
        \drawedge[curvedepth=-2,ELside=r,linecolor=blue](4,3){$a$}
        \drawedge[curvedepth=2,linecolor=blue](4,3){$c$}
      \end{picture}
    }
  \end{picture}
  \caption{A labeled in-split from $\A$ to $\B$.}
\label{BBEP:figureInSplit}
\end{figure}
}{%
\begin{figure}[hbt]
  \centering
  \begin{picture}(75,35)(-12,-8)
  \gasset{Nadjust=wh, linewidth=0.3}
  \gasset{fillcolor=lightgray!10}
  \put(0,-5){
      \begin{picture}(20,20)
        \node[linewidth=0.1,Nfill=n](1)(0,10){$1$}
        \node[dash={0.5 0.5}0](2)(20,10){$2$}
        \drawloop[linecolor=lightgray](1){$a$}
        \drawloop[loopangle=180,linecolor=lightgray](1){$c$}
        \drawedge[curvedepth=3,dash={0.5 0.5}0](1,2){$b$}
        \drawedge[curvedepth=3](2,1){$a$} 
      \end{picture}
    }
    \put(40,0){
      \begin{picture}(20,30)
        \node[linecolor=lightgray](3)(0,15){$3$}
        \node[dash={0.5 0.5}0](5)(20,5){$5$}
        \node(4)(0,-5){$4$}
        \drawloop[linecolor=lightgray](3){$a$}
        \drawloop[loopangle=180,linecolor=lightgray](3){$c$}
        \drawedge[curvedepth=3,dash={0.5 0.5}0](3,5){$b$}
        \drawedge[curvedepth=3](5,4){$a$}
        \drawedge[curvedepth=3,dash={0.5 0.5}0](4,5){$b$}
        \drawedge[curvedepth=-2,ELside=r,linecolor=lightgray](4,3){$a$}
        \drawedge[curvedepth=2,linecolor=lightgray](4,3){$c$}
      \end{picture}
    }
   \end{picture}
  \caption{A labeled in-split from $\A$ to $\B$.}\label{BBEP:figureInSplit}
\end{figure}
}
\begin{example}
  Let $\A$ and $\B$ be the automata represented on
  Figure~\ref{BBEP:figureInSplit}. Here $Q=\{1,2\}$ and $R=\{3,4,5\}$.
  One has $M(\A)=ND$ and $M(\B)=DN$ with
  \begin{displaymath}
    N=\begin{bmatrix}a+c&0&b\\0&a&0\end{bmatrix},\quad 
    D=\begin{bmatrix}1&0\\1&0\\0&1\end{bmatrix}.
  \end{displaymath}
\end{example}

A \emph{labeled out-merge}%
\index{labeled!out-merge}\index{out-merge!labeled}\index{automaton!labeled
  out-merge} from $\B$ onto $\A$ is an out-merge $(h,k)$ from $H$ onto
$G$ such that for each $f\in F$ the labels of $f$ and $h(f)$ are
equal.

We say that $\B$ is a \emph{labeled out-split}%
\index{labeled!out-split}\index{out-split!labeled}\index{automaton!labeled
  out-split} of $\A$, or that $\A$ is a \emph{labeled in-merge} of
$\B$.
 
Thus if $\B$ is a labeled out-split of $\A$, there is a labeled
conjugacy
from $X_\B$ onto $X_\A$.
\begin{proposition}
  The automaton $\B=(R,F)$ is a labeled out-split of the automaton
  $\A=(Q,E)$ if and only if there is a $Q\times R$-row division matrix
  $D$ and an alphabetic $R\times Q$-matrix $N$ such that
  \begin{equation}\label{EqOutSplitAutomata}
    M(\A)=DN\,,\quad M(\B)=ND\,.
  \end{equation}
\end{proposition}

\ifthenelse{\boolean{colorprint}}{%
\begin{figure}[hbt]
\centering
\gasset{Nw=5,Nh=5}
  \begin{picture}(75,30)(-10,-10)
\put(0,-5){
\begin{picture}(20,20)
\node(1)(0,10){$1$}
\node[linecolor=green](2)(20,10){$2$}
\drawloop[linecolor=blue](1){$a$}
\drawloop[loopangle=180,linecolor=red](1){$c$}
\drawedge[curvedepth=3,linecolor=blue](1,2){$b$}
\drawedge[curvedepth=3,linecolor=green](2,1){$a$}
\end{picture}
}
\put(40,0){
\begin{picture}(20,30)
\node[linecolor=blue](3)(0,15){$3$}
\node[linecolor=green](5)(20,5){$5$}
\node[linecolor=red](4)(0,-5){$4$}
\drawloop[loopangle=180,linecolor=blue](3){$a$}
\drawloop[loopangle=180,linecolor=red](4){$c$}
\drawedge[curvedepth=3,linecolor=blue](3,5){$b$}
\drawedge[curvedepth=3,linecolor=green](5,4){$a$}
\drawedge[curvedepth=3,linecolor=green](5,3){$a$}
\drawedge[curvedepth=2,linecolor=blue](3,4){$a$}
\drawedge[curvedepth=2,linecolor=red](4,3){$c$}
\end{picture}
}
\end{picture}
\caption{A labeled out-split from  $\A$ to $\B$.}\label{BBEP:figureOutSplit}
\end{figure}
}{
\begin{figure}[hbt]
  \centering
  \gasset{Nw=5,Nh=5}\gasset{linewidth=0.3}
  \begin{picture}(75,30)(-10,-10)
    \gasset{fillcolor=lightgray!10}
    \put(0,-5){
      \begin{picture}(20,20)
        \node[linewidth=0.1,Nfill=n](1)(0,10){$1$}
        \node[dash={0.5 0.5}0](2)(20,10){$2$}
        \drawloop[linecolor=lightgray](1){$a$}
        \drawloop[loopangle=180](1){$c$}
        \drawedge[curvedepth=3,linecolor=lightgray](1,2){$b$}
        \drawedge[curvedepth=3,dash={0.5 0.5}0](2,1){$a$}
      \end{picture}
    }
    \put(40,0){
      \begin{picture}(20,30)
        \node[linecolor=lightgray](3)(0,15){$3$}
        \node[dash={0.5 0.5}0](5)(20,5){$5$}
        \node(4)(0,-5){$4$}
        \drawloop[loopangle=180,linecolor=lightgray](3){$a$}
        \drawloop[loopangle=180](4){$c$}
        \drawedge[curvedepth=3,linecolor=lightgray](3,5){$b$}
        \drawedge[curvedepth=3,dash={0.5 0.5}0](5,4){$a$}
        \drawedge[curvedepth=3,dash={0.5 0.5}0](5,3){$a$}
        \drawedge[curvedepth=2,linecolor=lightgray](3,4){$a$}
        \drawedge[curvedepth=2](4,3){$c$}
      \end{picture}
    }
  \end{picture}
  \caption{A labeled out-split from  $\A$ to $\B$.}\label{BBEP:figureOutSplit}
\end{figure}
}
\begin{example}
  Let $\A$ and $\B$ be the automata represented on
  Figure~\ref{BBEP:figureOutSplit}. Here $Q=\{1,2\}$ and
  $R=\{3,4,5\}$.  One has $M(\A)=ND$ and $M(\B)=DN$ with
  \begin{displaymath}
    N=\begin{bmatrix}a&b\\c&0\\a&0\end{bmatrix},
    \quad D=\begin{bmatrix}1&1&0\\0&0&1\end{bmatrix}.
  \end{displaymath}
\end{example}

Let $\A=(Q,E)$ be an automaton. For a pair of integers $m,n\ge 0$,
denote by $\A^{[m,n]}$ the following automaton called the $(m,n)$-th
\emph{extension}\index{automaton!extension}\index{extension automaton}
of $\A$.  The underlying graph of $\A^{[m,n]}$ is the higher edge
graph $G^{[k]}$ for $k=m+n+1$. The label of an edge
\begin{displaymath}
  p_0\edge{a_1}p_1\edge{a_2}\cdots\edge{a_m}p_m
  \edge{a_{m+1}}p_{m+1}\edge{a_{m+2}}\cdots
  \edge{a_{m+n}}p_{m+n}\edge{a_{m+n+1}}p_{m+n+1}
\end{displaymath}
is the letter $a_{m+1}$.  Observe that $\A^{[0,0]}=\A$. By this
construction, each graph $G^{[k]}$ produces $k$ extensions according
to the choice of the labeling.

\begin{proposition}\label{stExtension}
  For $m\ge 1,n\ge 0$, the automaton $\A^{[m-1,n]}$ is a labeled
  in-merge of the automaton $\A^{[m,n]}$ and for $m\ge 0,n\ge 1$, the
  automaton $\A^{[m,n-1]}$ is a labeled out-merge of the automaton
  $\A^{[m,n]}$.
\end{proposition}
\begin{proof}
  Suppose that $m\ge 1,n\ge 0$. Let $k$ be the map from the paths of
  length $m+n$ in $\A$ onto the paths of length $m+n-1$ which erases
  the first edge of the path. Let $h$ be the map from the set of edges
  of $\A^{[m,n]}$ to the set of edges of $\A^{[m-1,n]}$ defined by
  $h(\pi,a,\rho)=(k(\pi),a,k(\rho))$.  Then $(h,k)$ is a labeled
  in-merge from $\A^{[m,n]}$ onto $\A^{[m-1,n]}$.  The proof that,
  for $m\ge 0,n\ge 1$, the automaton $\A^{[m,n-1]}$ is an out-merge of
  the automaton $\A^{[m,n]}$ is symmetrical.
\end{proof}

The following result is the analogue, for automata, of the
Decomposition Theorem.

\begin{theorem}\label{AutomataDecompositiontheorem}
  Every conjugacy of automata is a composition of
  labeled splits and merges.
\end{theorem}

\begin{proof} Let $\A$ and $\B$ be two conjugate automata.
  Let $\varphi$ be a labeled conjugacy from $\A$ onto $\B$. Let $G_0$
  and $H_0$ be the underlying graphs of $\A$ and $\B$, respectively.
  By the Decomposition Theorem~\ref{ShiftDecompositionTheorem}, there
  are sequences $(G_1,\ldots,G_n)$ and $(H_1,\ldots,H_m)$ of graphs
  with $G_n=H_m$ and such that $G_{i+1}$ is a split of $G_i$ for $0\le
  i<n$ and $H_{j+1}$ is a split of $H_j$ for $0\le j< m$. Moreover,
  $\varphi$ is the composition of the sequence of edge splitting maps from
  $G_i$ onto $G_{i+1}$ followed by the sequence of edge merging maps from
  $H_{j+1}$ onto $H_j$. Let $(h_i,k_i)$, for $1\le i\le n$, be a merge
  from $G_{i}$ onto $G_{i-1}$ and $(u_j,v_j)$, for $1\le j\le m$ be a
  merge from $H_{j}$ onto $H_{j-1}$. Then we may define labels on the
  edges of $G_1,\ldots,G_n$ in such a way that $G_i$ becomes the
  underlying graph of an automaton $\A_i$ and $(h_i,k_i)$ is a labeled
  merge from $\A_{i}$ onto $\A_{i-1}$.  In the same way, we may define
  labels on the edges of $H_j$ in such a way that $H_j$ becomes the
  underlying graph of an automaton $\B_j$ and $(u_j,v_j)$ is a labeled
  merge from $\B_j$ onto $\B_{j-1}$.
  \begin{displaymath}
    G_0\xleftarrow{(h_1,k_1)}G_1\cdots\xleftarrow{(h_n,k_n)}G_n
    =H_m\edge{(u_m,v_m)}\cdots H_1\edge{(u_1,v_1)}H_0\,.
  \end{displaymath}
  Let $h=h_1\cdots h_n$ and $u=u_1u_2\cdots u_m$. Since
  $\varphi=u_\infty h_{\infty}^{-1}$, and $\varphi$ is a labeled
  conjugacy, we have $\lambda_\A h_\infty=\lambda_\B u_\infty$.  This
  shows that the automata $\A_n$ and $\B_m$ are equal. Thus there is a
  sequence of labeled splitting maps followed by a sequence of labeled
  merging maps which is a equal to $\varphi$.
\end{proof}

Let $M$ and $M'$ be two alphabetic square matrices over the same
alphabet $A$.  We say that $M$ and $M'$ are \emph{elementary
  equivalent}%
\index{elementary equivalent matrices}%
\index{matrix!elementary equivalent}%
\index{equivalent matrices!elementary} if there exists a nonnegative
integral matrix $D$ and an alphabetic matrix $N$ such that
\begin{displaymath}
  M=DN\,,\quad M'=ND\quad \text{or vice-versa}.
\end{displaymath}

By Proposition~\ref{defEquivLabeledInSplit}, if $\B$ is an in-split of
$\A$, then $M(\B)$ and $M(\A)$ are elementary equivalent.  We say that
$M,M'$ are \emph{strong shift equivalent}%
\index{strong shift equivalent matrices}%
\index{matrix!strong shift equivalent}%
\index{equivalent matrices!strong shift} if there is a sequence
$(M_0,M_1,\ldots,M_n)$ such that $M_i$ and $M_{i+1}$ are elementary
equivalent for $0\le i<n$ with $M_0=M$ and $M_n=M'$.  The following
result is the version, for automata, of the Classification Theorem.

\begin{theorem}\label{AutomataClassificationtheorem}
  Two automata are conjugate if and only if their adjacency matrices
  are strong shift equivalent.
\end{theorem}

Note that when $D$ is a column division matrix, the statement results
from Propositions~\ref{stLabeledInMergeIsCojugacy} and
\ref{defEquivInMerge}.  The following statement proves the theorem in
one direction.

\begin{proposition}\label{BBEP:PropInSplit}
  Let $\A$ and $\B$ be two automata. If $M(\A)$ is elementary
  equivalent to $M(\B)$, then $\A$ and $\B$ are conjugate.
\end{proposition}

\begin{proof} Let $\A=(Q,E)$ and $\B=(R,F)$. Let $D$ be an $R\times Q$
  nonnegative integral matrix and let $N$ be an alphabetic $Q\times R$
  matrix such that
\begin{displaymath}
M(\A)=ND,\quad M(\B)=DN.
\end{displaymath}
  Consider the map $f$ from the set of paths of length $2$ in $\A$ into
  $F$ defined as follows (see Figure~\ref{Figfg} on the left).  Let $p\edge{a}q\edge{b}r$ be a path of
  length $2$ in $\A$. Since $(M(\A)_{pq},a)=1$ and $M(\A)=ND$ there is a
  unique $t\in R$ such that $(N_{pt},a)=D_{tq}=1$. In the same way,
  since $(M(\A)_{qr},b)=1$, there is a unique $u\in R$ such that
  $(N_{qu},b)=D_{ur}=1$. Since $M(\B)=DN$, we have
  $(M(\B)_{tu},b)=D_{tq}=(N_{qu},b)=1$ and thus $(t,u,b)$ is an edge of
  $\B$. We set
\begin{displaymath}
  f(p\edge{a}q\edge{b}r)=t\edge{b}u
\end{displaymath}
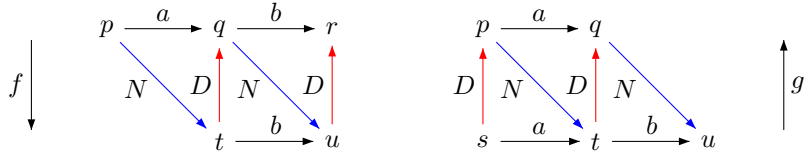
\begin{figure}[hbt]
  \centering
  \gasset{Nframe=n,Nadjust=wh}
\begin{picture}(100,20)
\put(0,0){
  \begin{picture}(60,20)(-15,0)
\node(fh)(-10,15){}\node(fb)(-10,0){}\drawedge[ELside=r](fh,fb){$f$}
    \node(p)(0,15){$p$}\node(q)(15,15){$q$}\node(r)(30,15){$r$}
    \node(t)(15,0){$t$}\node(u)(30,0){$u$}
    \drawedge(p,q){$a$}\drawedge(q,r){$b$}
    \drawedge[linecolor=red](t,q){$D$}\drawedge[linecolor=red](u,r){$D$}
    \drawedge[ELside=r,ELpos=40,linecolor=blue](p,t){$N$}\drawedge[ELside=r,ELpos=40,linecolor=blue](q,u){$N$}
    \drawedge(t,u){$b$}
  \end{picture}
}
\put(50,0){
  \begin{picture}(60,20)(-15,0)
    \node(gh)(40,15){}\node(gb)(40,0){}\drawedge[ELside=r](gb,gh){$g$}
    \node(p)(0,15){$p$}\node(q)(15,15){$q$}
    \node(s)(0,0){$s$}\node(t)(15,0){$t$}\node(u)(30,0){$u$}
    \drawedge(p,q){$a$}
    \drawedge[linecolor=red](s,p){$D$}\drawedge[linecolor=red](t,q){$D$}
    \drawedge[ELside=r,ELpos=40,linecolor=blue](p,t){$N$}\drawedge[ELside=r,ELpos=40,linecolor=blue](q,u){$N$}
    \drawedge(s,t){$a$}\drawedge(t,u){$b$}
  \end{picture}
}
\end{picture}
\caption{The maps $f$ and $g$.}\label{Figfg}
\end{figure}
Similarly, we may define a map $g$ from the set of paths of length $2$ in
$\B$
into $E$ by
\begin{displaymath}
  g(s\edge{a}t\edge{b}u)=p\edge{a}q
\end{displaymath}
if $D_{sp}=(N_{pt},a)=D_{tq}=1$. Let $\varphi=f^{[1,0]}_\infty$ and
$\gamma=g^{[0,1]}_\infty$ (see Figure~\ref{Figfg} on the right). We
verify that
\begin{displaymath}
  \varphi\gamma=\Id_F,\quad \gamma\varphi=\Id_E
\end{displaymath}
where $\Id_E$ and $\Id_F$ are the identities on $E^\Z$ and $F^\Z$.
Let indeed $\pi$ be a path in $X_\A$ and let $\rho=\varphi(\pi)$. Set
$\pi_i=(p_i,a_i,p_{i+1})$
and $\rho_i=(r_i,b_i,r_{i+1})$ (see Figure~\ref{BBEP:fig:strong}). Then, by definition of $\varphi$, we
have for all $i\in\Z$, $b_{i}=a_{i}$ and
$(N_{p_ir_{i+1}},a_i)=D_{r_ip_i}=1$.
Let $\sigma=\gamma(\rho)$ and $\sigma=(s_i,c_i,s_{i+1})$. By
definition of $\gamma$, we have $c_i=b_i$ and
$D_{r_is_i}=(N_{s_ir_{i+1}},b_i)=1$.
Thus we have simultaneously $D_{r_ip_i}=(N_{p_ir_{i+1}},a_i)=1$ and
$D_{r_is_i}=(N_{s_ir_{i+1}},a_i)=1$. Since $M(\A)=DN$, this forces
$p_i=s_i$. Thus $\sigma=\pi$ and this shows that
$\gamma\varphi=\Id_E$.
The fact that $\varphi\gamma=\Id_F$ is proved in the same way.

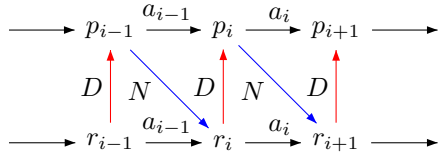
\begin{figure}[hbt]
  \centering
  \gasset{Nframe=n,Nadjust=wh}
  \begin{picture}(60,20)(-15,0)
    \node(-)(-15,15){}\node(p)(0,15){$p_{i-1}$}\node(q)(15,15){$p_i$}\node(r)(30,15){$p_{i+1}$}\node(+)(45,15){}
    \node(--)(-15,0){}\node(s)(0,0){$r_{i-1}$}\node(t)(15,0){$r_i$}\node(u)(30,0){$r_{i+1}$}\node(++)(45,0){}
    \drawedge(-,p){}\drawedge(p,q){$a_{i-1}$}\drawedge(q,r){$a_i$}\drawedge(r,+){}
    \drawedge[linecolor=red](s,p){$D$}\drawedge[linecolor=red](t,q){$D$}\drawedge[linecolor=red](u,r){$D$}
    \drawedge[ELside=r,ELpos=40,linecolor=blue](p,t){$N$}\drawedge[ELside=r,ELpos=40,linecolor=blue](q,u){$N$}
    \drawedge(--,s){}\drawedge(s,t){$a_{i-1}$}\drawedge(t,u){$a_i$}\drawedge(u,++){}
  \end{picture}
\caption{Conjugacy of automata.}\label{BBEP:fig:strong}
\end{figure}
\end{proof}

\begin{proof}[Proof of Theorem~\ref{AutomataClassificationtheorem}]
  In one direction, the above statement is a direct consequence of the
  Decomposition Theorem~\ref{ShiftDecompositionTheorem}. Indeed, if
  $\A$ and $\B$ are conjugate, there is a sequence
  $\A_0,\A_1,\ldots,\A_n$ of automata such that $\A_i$ is a split or a
  merge of $\A_{i+1}$ for $0\le i<n$ with $\A_0=\A$ and $\A_n=\B$. The
  other direction follows from Proposition~\ref{BBEP:PropInSplit}.
\end{proof}

%
%
\section{Minimal automata}\label{sectionMinimalAutomata}

In this section, we define two notions of minimal automaton for sofic
shifts: the Krieger automaton and the Fischer automaton. The first is
defined for any sofic shift, and the second for irreducible ones.

The main result is that the Fischer automaton has the minimal number
of states among all deterministic automata recognizing a given sofic shift
(Proposition~\ref{propReduction}).

We then define the syntactic semigroup of a sofic shift, as an
ordered semigroup. We show that this semigroup is isomorphic to the 
transition semigroup of the Krieger automaton and, for irreducible
shifts, to the transition semigroup of the Fischer automaton
(Proposition~\ref{BBP:syntacticSemigroup}). 

\intertitre{Minimal automata of sets of finite words} Recall that an
automaton $\A=(Q,E)$ recognizes a shift $X$ if $X=L_\A$. There should
be no confusion with the notion of acceptance for sets of finite words
in the usual sense: if $\A$ has an initial state $i$ and a set of
terminal states $T$, the set of finite words recognized by $\A$ is the
set of labels of finite paths from $i$ to a terminal state $t$ in
$T$. In this chapter\footnote{This contrasts the more traditional
  definition which assumes in addition that there is a unique initial
  state.}, an automaton is called
\emph{deterministic}\index{automaton!deterministic}\index{deterministic
  automaton} if, for each state $p$ and each letter $a$, there is at
most one edge starting in $p$ and carrying the label $a$. We write, as
usual, $p\cdot u$ for the unique end state, provided it exists, of a
path starting in $p$ and labeled $u$.  For a set $W$ of $A^*$, there
exists a unique deterministic minimal
automaton\index{automaton!minimal}\index{minimal!automaton} (this time
with a unique initial state) recognizing $W$. Its states are the
nonempty sets $u^{-1}W$ for $u\in A^*$, called the \emph{right
  contexts}\index{context, right}\index{right!context} of $u$, and the
edges are the triples $(u^{-1}W,a,(ua)^{-1}W)$, for $a\in A$ (see the
chapter of J.-\'E.  Pin).


Let $\A=(Q,E)$ be a finite automaton. For a state $p\in Q$, we denote
by $L_p(\A)$ or simply $L_p$
the set of labels of finite paths starting from $p$.  The
automaton $\A$ is said to be \emph{reduced}%
\index{reduced automaton}\index{automaton!reduced} if $p\ne q$ implies
$L_p\ne L_q$.

A word $w$ is \emph{synchronizing}\index{synchronizing word} for a
deterministic automaton $\A$ if the set of paths labeled $w$ is
nonempty
and all paths labeled $w$ end in the same
state. An automaton is \emph{synchronized}%
\index{synchronized automaton}\index{automaton!synchronized} if there
is a synchronizing word. The following result holds because all
states are terminal.

\begin{proposition}\label{propRedIsSync}
  A reduced deterministic automaton is synchronized.
\end{proposition}

\begin{proof}
  Let $\A=(Q,E)$ be a reduced deterministic automaton.  Given any word
  $x$, we denote by $Q\cdot X$ the set $Q\cdot x=\{q\cdot x\mid q\in
  Q\}$.

  Let $x$ be a word such that $Q\cdot x$ has minimal nonzero
  cardinality. Let $p,q$ be two elements of the set $Q\cdot x$. If $u$
  is a word such that $p\cdot u$ is nonempty, then $q\cdot u$ is also
  nonempty since otherwise $Q\cdot xu$ would be of nonzero cardinality
  less than $Q\cdot x$.  This implies that $L_p=L_q$ and thus $p=q$
  since $\A$ is reduced.  Thus $x$ is synchronizing.
\end{proof}

\subsection{Krieger automata and Fischer automata}
\intertitre{Krieger automata}

We denote by $A^{-\N}$ the set of left infinite words $x=\cdots
x_{-1}x_0$. For $y=\cdots y_{-1}y_0\in A^{-\N}$ and $z=z_0z_1\cdots\in
A^{\N}$, we denote by $y\cdot z=(w_i)_{i\in\Z}$ the biinfinite word defined
by $w_i=y_{i+1}$ for $i<0$ and $w_i=z_i$ for $i\ge 0$.  Let $X$ be a
shift space. For $y\in A^{-\N}$, the set of \emph{right
  contexts}\index{context, right}\index{right!context} of $y$ is the
set $C_X(y)=\{z\in A^{\N}\mid y\cdot z\in X\}$. For $u\in A^+$, we denote
$u^\omega=uu\cdots$.

The \emph{Krieger automaton}\index{Krieger
  automaton}\index{automaton!Krieger} of a shift space $X$ is the
deterministic automaton whose states are the nonempty sets of the form $C_X(y)$ for
$y\in A^{-\N}$, and whose edges are the triples $(p,a,q)$ where
$p=C_X(y)$ for some left infinite word, $a\in A$ and $q=C_X(ya)$.

The definition of the Krieger automaton uses infinite words. One could
use instead of
the sets $C_X(y)$ for $y\in A^{-\N}$, the sets 
\begin{displaymath}
D_X(y)=\{u\in A^*\mid \exists z\in A^{\N}: yuz\in X\}.
\end{displaymath}
Indeed $C_X(y)=C_X(y')$ if and only if $D_X(y)=D_X(y')$. However, one
cannot
dispense completely with infinite words (see Proposition~\ref{propKriegerReduced}).
\begin{example}
  Let $A=\{a,b\}$, and let $X=X^{(ba)}$.  The Krieger automaton of $X$
  is represented in Figure~\ref{FigKrieger}. The states are the sets
  $1=C_X(\cdots aaa)=a^\omega\cup a^*b^\omega$ and $2=C_X(\cdots
  aaab)=b^\omega$.
\end{example}

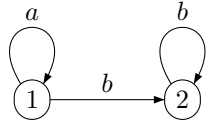
\begin{figure}[hbt]
\centering
\begin{picture}(20,15)(0,-3)
\gasset{Nadjust=wh}
\node(1)(0,0){$1$}\node(2)(20,0){$2$}
\drawloop[loopangle=90](1){$a$}\drawedge(1,2){$b$}\drawloop[loopangle=90](2){$b$}
\end{picture}
\caption{The Krieger automaton of $X^{(ba)}$.}\label{FigKrieger}
\end{figure}

\begin{proposition}\label{propKriegerReduced}
  The Krieger automaton of a shift space $X$ is reduced and recognizes
  $X$. It is finite if and only if $X$ is sofic.
\end{proposition}
\begin{proof}
  Let $\A=(Q,E)$ be the Krieger automaton of $X$.  Let $p,q\in Q$ and
  let $y,z\in A^{-\N}$ be such that $p=C_X(y)$, $q=C_X(z)$. If
  $L_p=L_q$, then the labels of infinite paths starting from $p$ and
  $q$ are the same. Thus $p=q$. This shows that $\A$ is reduced. If
  $\A$ finite, then $X$ is sofic by
  Proposition~\ref{BBEP:PropSoficRec}.  Conversely, if $X$ is sofic,
  let $\A$ be a finite automaton recognizing $X$.  The set of right
  contexts of a left infinite word $y$ only depends on the set of
  states $p$ such that there is a path in the automaton $\A$ labeled
  $y$ ending in state $p$. Thus the family of sets of right contexts
  is finite.
\end{proof}
We say that a deterministic automaton $\A=(Q,E)$ over the alphabet $A$ is a
\emph{subautomaton} of a deterministic automaton\index{subautomaton}
$\A'=(Q',E')$ if $Q\subset Q'$ and if for each edge $(p,a,q)\in E$
such that $p\in Q$ one has $q\in Q$ and $(p,a,q)\in E'$.

The following proposition appears in \cite{Nasu:1988} and
 in~\cite{Costa2007b}
where an algorithm to compute the states of the minimal automaton
which are in the Krieger automaton is described. 

\begin{proposition}\label{propAlfredo}
  The Krieger automaton of a sofic shift $X$ is,  up to an
  isomorphism, a subautomaton of the minimal automaton of the set of blocks of $X$.
\end{proposition}

\begin{proof}
Let $X$ be a sofic shift. Let $y\in A^{-\N}$ and set $y=\cdots
y_{-1}y_0$ with $y_i\in A$ for $i\le 0$. Set $u_i=y_{-i}\cdots y_0$ and
$U_i=u_i^{-1}\B(X)$. Since $\B(X)$ is regular, the chain
\begin{displaymath}
\ldots \subset U_i\subset\ldots\subset U_1\subset U_0
\end{displaymath}
is stationary. Thus there is an integer $n\ge 0$ such that
$U_{n+i}=U_n$ for all $i\ge 0$. We define $s(y)=U_n$.

We show that the map $C_X(y)\mapsto s(y)$ is well-defined and injective.
Suppose first that $C_X(y)=C_X(y')$ for some $y,y'\in A^{-\N}$. Let
$u\in A^*$ be such that $y_{-m}\cdots y_0u\in \B(X)$ for all $m\ge n$. By compactness, there exists
a $z\in A^\N$ such that $yuz\in X$. Then $y'\cdot uz\in X$ implies
$u\in s(y')$. Symmetrically $u\in s(y')$ implies $u\in s(y)$.
This shows that the map is well-defined.

To show that it is injective, consider $y,y'\in A^{-N}$ such that
$s(y)=s(y')$. Let $z\in C_X(y)$. For each integer $m\ge 0$, we have
$z_0\cdots z_m\in s(y)$ and thus $z_0\cdots z_m\in s(y')$. Since $X$
is closed, this implies that $y'\cdot z\in X$ and thus $z\in C_X(y')$.
The converse implication is proved in the same way.
\end{proof}

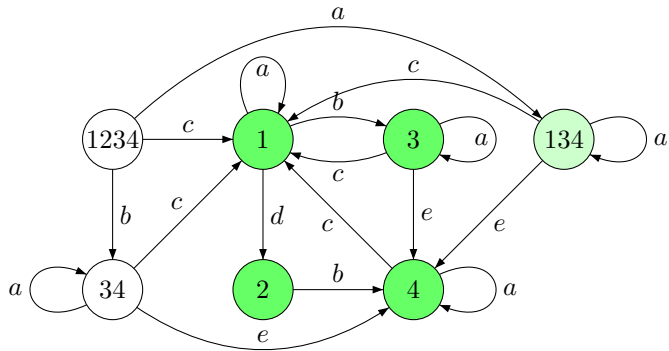
\begin{figure}[hbt]
\centering
\gasset{fillcolor=green!60}
\begin{picture}(80,43)(0,-7)
\gasset{Nfill=n}
\node(1234)(0,20){$1234$}\node(34)(0,0){$34$}
\gasset{Nfill=y}
\node(1)(20,20){$1$}\node(3)(40,20){$3$}\node[fillcolor=green!20](134)(60,20){$134$}
\node(2)(20,0){$2$}\node(4)(40,0){$4$}
\drawedge[curvedepth=15](1234,134){$a$}\drawedge(1234,1){$c$}\drawedge(1234,34){$b$}
\drawloop[loopangle=90,ELside=r](1){$a$}\drawedge[curvedepth=3](1,3){$b$}\drawedge(1,2){$d$}
\drawloop[loopangle=0,ELside=r](3){$a$}\drawedge[curvedepth=3](3,1){$c$}\drawedge(3,4){$e$}
\drawloop[loopangle=0](134){$a$}\drawedge(134,4){$e$}\drawedge[curvedepth=-8,ELside=r](134,1){$c$}
\drawloop[loopangle=180](34){$a$}\drawedge[curvedepth=-8](34,4){$e$}\drawedge(34,1){$c$}
\drawedge(2,4){$b$}
\drawloop[loopangle=0](4){$a$}\drawedge(4,1){$c$}
\end{picture}
\caption{An example of  Krieger automaton.}\label{exNasu}
\end{figure}
\begin{example}\label{exAlfredo}
Consider the automaton on $7$ states given in Figure~\ref{exNasu}. It
is obtained, starting with the subautomaton over  the states $1,2,3,4$,
using the subset construction computing the accessible nonempty sets of states,
starting from the set $\{1,2,3,4\}$.

The subautomaton with dark shaded states $1,2,3,4$ is strongly
connected and recognizes an irreducible sofic shift denoted by $X$.
The whole automaton is the minimal automaton (with initial state
$\{1,2,3,4\}$) of the set of blocks of $X$. The Krieger automaton of
$X$ is the automaton on the five shaded states.  Indeed, with the
notation of the proof, there is no left infinite word $y$ such that
$s(y)=\{1,2,3,4\}$ or $s(y)=\{3,4\}$.
\end{example}

\intertitre{Fischer automata of irreducible shift spaces}

A shift space $X\subset A^\Z$ is called
\emph{irreducible}\index{irreducible shift}\index{shift space!²irreducible}
if for any $u,v\in\B(X)$ there exists a $w\in\B(X)$ such that $uwv\in
\B(X)$.  

An automaton is said to be strongly connected if its underlying graph
is
strongly connected.
Clearly a shift recognized by a strongly connected automaton is irreducible.

A strongly connected component of an automaton $\A$ is \emph{minimal}%
\index{minimal!strongly connected component}\index{strongly!connected
  component, minimal} if all successors of vertices of the component
are themselves in the component. One may verify that a minimal
strongly
connected component is the same as a strongly connected subautomaton.

The following result is due to Fischer~\cite{Fischer:1975}
(see also \cite[Section 3]{Lind&Marcus:1995}). It implies in
particular that an irreducible sofic shift can be recognized by
a strongly connected automaton.

\begin{proposition}\label{PropFischer}
  The Krieger automaton of an irreducible sofic shift $X$ is
  synchronized and has a unique minimal strongly connected component.
\end{proposition}

\begin{proof}
  Let $\A=(Q,E)$ be the Krieger automaton of $X$. By
  Proposition~\ref{propKriegerReduced},
$\A$ is reduced and by Proposition~\ref{propRedIsSync}, it follows that it is synchronized.

  Let $x$ be a synchronizing word.
Let $R$ be the set of states reachable from the state $q=Q\cdot
  x$. The set $R$ is a minimal strongly connected component of $\A$. Indeed, for
  any $r\in R$ there is a path $q\edge{y}r$. Since $X$ is irreducible
  there is a word $z$ such that $yzx\in \B(X)$. Since $q\cdot yzx=q$,
  $r$ belongs to the same strongly connected component as $q$. Next,
  if $p$ belongs to a minimal strongly connected component $S$ of
  $\A$, since $X$ is irreducible, there is a word $y$ such that
  $p\cdot yx$ is not empty.  Thus $q$ is in $S$, which implies
  $S=R$. Thus $R$ is the only minimal strongly component of $\A$.
\end{proof}

\begin{example}\label{ExMinAutomaton}
  Let $X$ be the even shift.  The Krieger and Fischer automata of $X$
  are represented on Figure~\ref{figMinAutomatonEven}.  The word $a$ is
  synchronizing.
\end{example}

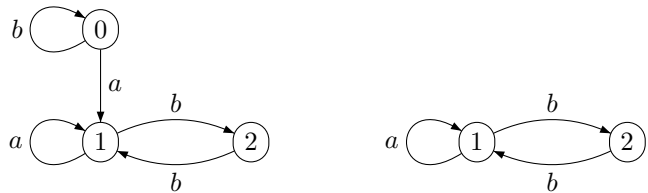
\begin{figure}[hbt]
  \centering\gasset{Nadjust=wh}
  \begin{picture}(70,20)
    \put(0,0){
      \begin{picture}(20,20)(0,-3)
        \node(0)(0,15){$0$}
        \node(1)(0,0){$1$}\node(2)(20,0){$2$}
        \drawloop[loopangle=180](0){$b$}\drawedge(0,1){$a$}
        \drawloop[loopangle=180](1){$a$}\drawedge[curvedepth=3](1,2){$b$}
        \drawedge[curvedepth=3](2,1){$b$}
      \end{picture}
    }
    \put(50,0){
      \begin{picture}(20,10)(0,-3)
        \node(1)(0,0){$1$}\node(2)(20,0){$2$}
        \drawloop[loopangle=180](1){$a$}\drawedge[curvedepth=3](1,2){$b$}
        \drawedge[curvedepth=3](2,1){$b$}
      \end{picture}
    }
  \end{picture}
  \caption{The Krieger and Fischer automata of $X$.}\label{figMinAutomatonEven}
\end{figure}

\begin{example}
The Fischer automaton of the irreducible shift of
Example~\ref{exAlfredo}
is the subautomaton on states $1,2,3,4$ represented with dark shaded
states
in Figure~\ref{exNasu}.
\end{example}

Let $X$ be an irreducible sofic shift $X$.  The minimal strongly
connected component of the Krieger automaton of $X$ is called its
\emph{Fischer automaton}%
\index{Fischer automaton}\index{automaton!Fischer}.

\begin{proposition}
  The Fischer automaton of an irreducible sofic shift $X$ recognizes $X$.
\end{proposition}
\begin{proof}
  The Fischer automaton $\F$ of $X$ is a subautomaton of the Krieger
  automaton of $X$ which in turn is a subautomaton of the minimal
  automaton $\A$ of the set $\B(X)$. Let $i$ be the initial state of
  $\A$. Since $\A$ is trim, there is a word $w$ such that $i\cdot w$
  is a state of $\F$. Let $v$ be any block of
  $X$. Since $X$ is irreducible, there is a word $u$ such that $wuv$
  is a block of $X$. This shows that $v$ is a label of a path in
  $\F$. Thus every block of $X$ is a label of a path in $\F$ and
  conversely. In view of Proposition~\ref{BBEP:propAutomatonSofic0},
  the automaton $\F$ recognizes $X$.
\end{proof}

Let $\A=(Q,E)$ and $B=(R,F)$ be two deterministic automata. A
\emph{reduction}\index{automaton!reduction}%
\index{reduction of automata} from $\A$ onto $\B$ is a map $h$ from
$Q$ onto $R$ such that for any letter $a\in A$, one has $(p,a,q)\in E$
if and only if $(h(p),a,h(q))\in F$. Thus any labeled in or out-merge
is a reduction. However the converse is not true since a reduction is
not, in general, a conjugacy.

For any automaton $\A=(Q,E)$, there is reduction from $\A$ onto a reduced
automaton
$\B$. It is obtained by identifying the pairs of states $p,q\in Q$
such that $L_p=L_q$.

The following statement is Corollary 3.3.20 of~\cite{Lind&Marcus:1995}.
\begin{proposition}\label{propReduction}
  Let $X$ be an irreducible shift space. For any strongly connected deterministic
  automaton $\A$ recognizing $X$ there is a reduction from $\A$ onto
  the Fischer automaton of $X$.
\end{proposition}

\begin{proof}
  Let $\A=(Q,E)$ be a strongly connected automaton recognizing $X$.
  Let $\B=(R,F)$ be the reduced automaton obtained from $\A$
  identifying the pairs $p,q\in Q$ such that $L_p=L_q$. By
  Proposition~\ref{propRedIsSync}, $\B$ is synchronized.

We now show that $\B$ can be identified with the Fischer automaton of
$X$.  Let $w$ be a synchronizing word for $\B$. Set $s=Q\cdot w$. Let
$r$ be a state such that $r\cdot w=s$.  and let $y\in A^{-\N}$ be the
label of a left infinite path ending in the state $s$. For any state
$t$ in $R$, let $u$ be a word such that $s\cdot u=t$. The set
$C_X(ywu)$ depends only on the state $t$, and not on the word $u$ such
that $s\cdot u=t$.  Indeed, for each right infinite word $z$, one has
$ywuz$ in $X$ if and only if there is a path labeled $z$ starting at
$t$. This holds because $w$ is synchronizing.

Thus the map $t\mapsto C_X(ywu)$ is well-defined and defines a
reduction from $\B$ onto the Fischer automaton of $X$.
\end{proof}

This statement shows that the Fischer automaton of an irreducible
shift $X$ is minimal\index{automaton!minimal}\index{minimal!automaton} in the sense that it has the minimal number of states
among all deterministic strongly connected automata recognizing $X$.

The statement also gives the following  practical method to compute the Fischer
automaton of an irreducible shift. We start with a strongly connected
deterministic automaton recognizing $X$ and merge the pairs of states
$p,q$ such that $L_p=L_q$. By the above result, the resulting automaton is the
Fischer automaton of $X$.

\subsection{Syntactic semigroup} \label{subSectionsyntacticSemigroup}
Recall that a
\emph{preorder}\index{preorder} on a set is a relation which is
reflexive and transitive. The equivalence associated to a preorder is
the equivalence relation defined by $u\equiv v$ if and only if $u\le
v$ and $v\le u$.  

Let $S$ be a semigroup. A preorder on $S$ is said to be
\emph{stable}\index{stable preorder} if $s\le s'$ implies $us\le us'$
and $su\le s'u$ for all $s,s',u\in S$.  An \emph{ordered
  semigroup}\index{ordered
  semigroup}\index{semigroup!ordered}\index{preorder!stable} $S$ is a
semigroup equipped with a stable preorder.  Any semigroup can be
considered as an ordered semigroup equipped with the equality order.

A \emph{congruence}\index{congruence} in an ordered semigroup $S$ is
the equivalence associated to a stable preorder which is coarser than
the preorder of $S$. The quotient of an ordered semigroup by a
congruence is the ordered semigroup formed by the classes of the
congruence.

The \emph{set of contexts}\index{context!of a word} of a word $u$ with
respect to a set $W\subset A^+$ is the set $\Gamma_W(u)$ of pairs of
words defined by $\Gamma_W(u)=\{(\ell,r)\in A^*\times A^*\mid \ell u
r\in W\}$.  The preorder on $A^+$ defined by $u\le_W v$ if
$\Gamma_W(u)\subset \Gamma_W(v)$ is stable and thus defines a
congruence of the semigroup $A^+$ equipped with the equality order
called the \emph{syntactic congruence}%
\index{syntactic!congruence}\index{congruence!syntactic}.
The \emph{syntactic semigroup}%
\index{syntactic!semigroup}\index{semigroup!syntactic} of a set
$W\subset A^*$ is the quotient of the semigroup $A^+$ by the syntactic
congruence.

Let $\A=(Q,E)$ be a deterministic automaton on the alphabet $A$.
Recall that for $p\in Q$ and $u\in A^+$, there is at most one path $\pi$
labeled $u$ starting in $p$. We set $p\cdot u=q$ if $q$ is
the end of $\pi$ and $p\cdot u=\emptyset$ if $\pi$ does not exist.  The
preorder defined on $A^+$ by $u\le_\A v$ if $p\cdot u\subset
p\cdot v$ for all $p\in Q$ is stable. The quotient of $A^+$ by the
congruence associated to this preorder is the \emph{transition
  semigroup}\index{semigroup!transition}\index{transition!semigroup}
of $\A$. 

The following property is standard, see the chapter of J.-\'E~Pin.
\begin{proposition}\label{propSyntIsTrans}
  The syntactic semigroup of a set $W\subset A^+$ is isomorphic to the
  transition semigroup of the minimal automaton of $W$.
\end{proposition}

The \emph{syntactic semigroup}\index{syntactic
  semigroup}\index{semigroup!syntactic} of a shift space $X$ is by
definition the syntactic semigroup of $\B(X)$.

\begin{proposition}\label{BBP:syntacticSemigroup}
  Let $X$ be a sofic shift and let $S$ be its syntactic semigroup.
  The transition semigroup of the Krieger automaton of $X$ is
  isomorphic to $S$. Moreover, if $X$ is irreducible, then it is
  isomorphic to the transition semigroup of its Fischer automaton.
\end{proposition}

\begin{proof}
  Let $\A$ be the minimal automaton of $\B(X)$, and let $\K$ be the
  Krieger automaton of $X$. We have to show that for any $u,v\in A^+$,
  one has $u\le_\A v$ if and only if $u\le_{\K}v$.  Since, by
  Proposition~\ref{propAlfredo}, $\K$ is isomorphic to a subautomaton
  of $\A$, the direct implication is clear. Indeed, if $p$ is a state
  of $\K$, then $L_p(\K)$ is equal to the set $L_p(\A)$. Consequently,
  if $u\le_\A v$ then $u\le_\K v$.  Conversely, suppose that
  $u\le_{\K}v$. We prove that $u\le_{\B(X)} v$. For this, let
  $(\ell,r)\in \Gamma_{\B(X)}(u)$. Then $\ell ur\in\B(X)$.  Then
  $y\cdot\ell urz\in X$ for some $y\in A^{-\N}$ and $z\in A^{\N}$. But
  since $C_X(y\ell u)\subset C_X(y\ell v)$, this implies $rz\in C_X(y\ell v)$
  and thus $\ell vr\in\B(X)$. Thus $u\le_{\B(X)} v$ which implies
  $u\le_\A v$.

Next, suppose that $X$ is irreducible. We have to show that $u\le_\A v$ if and only if
$u\le_{\F(X)}v$. Since
$\F(X)$ is a subautomaton of $\K(X)$ and $\K(X)$ is a subautomaton of
 $\A$, the direct implication is
clear. Conversely, assume that
$u\le_{\F(X)}v$. Suppose that $\ell u r\in \B(X)$.
Let $i$ be the initial state of $\A$ and let
$w$ be such that $i\cdot w$ is a state of $\F(X)$. Since $X$ is
irreducible, there is a word $s$ such that $ws\ell ur\in\B(X)$. But
then $i\cdot ws\ell ur\ne\emptyset$ implies $i\cdot ws\ell
vr\ne\emptyset$.
Thus $\ell vr\in\B(X)$. This shows that
 $u\le_{\B(X)}v$ and thus $u\le_\A v$.
\end{proof}

%
%

\section{Symbolic conjugacy}\label{sectionSymbolicConjugacy}

This section is concerned with a new notion of conjugacy between
automata called symbolic conjugacy. It extends the notion of labeled
conjugacy and captures the fact that the automata may be over
different alphabets. The table below summarizes the various notions.

\begin{center}
\begin{tabular}[c]{|l|l|l|}\hline
object type&isomorphism&elementary transformation\\\hline
shift spaces&conjugacy&split/merge\\
edge shifts&conjugacy&edge split/merge\\
integer matrices&strong shift equivalence&elementary equivalence\\
automata (same alphabet)&labeled conjugacy&labeled split/merge\\
automata &symbolic conjugacy&split/merge\\
alphabetic matrices&symbolic strong shift &elementary
symbolic \\\hline
\end{tabular}
\end{center}

There are two main results in this section. Theorem~\ref{TheoremNasu}
due to Nasu is a version of the Classification Theorem for sofic
shifts. It implies in particular that conjugate sofic shifts have
symbolic conjugate Krieger or Fisher automata.The proof uses the
notion of bipartite automaton, which corresponds to the symbolic
elementary equivalence of adjacency matrices.
Theorem~\ref{TheoremHamachiNasu} is due to Hamachi and Nasu: it
characterizes symbolic conjugate automata by means of their adjacency
matrices.

In this section, we will use for convenience automata in which several
edges with the same source and target can have the same label. Formally,
such an automaton is a pair $\A=(G,\lambda)$ of a graph $G=(Q,\E)$ and a map
assigning to each edge $e\in\E$ of  a label $\lambda(e)\in A$. The adjacency
matrix of $\A$ is the $Q\times Q$-matrix $M(\A)$ with elements in
$\N\langle A\rangle$ defined
by 
\begin{equation}
(M(\A)_{pq},a)=\Card\{e\in\E\mid \lambda(e)=a\}.
\end{equation} 
Note that $M(\A)$ is alphabetic but may have arbitrary nonnegative
coefficients.
The advantage of this version of automata is that for any alphabetic
$Q\times Q$-matrix $M$ there is an automaton $\A$ such that $M(\A)=M$.

We still denote by $X_\A$ the edge shift $X_G$ and by $L_\A$ the set
of labels of infinite paths in $G$.
\intertitre{Symbolic conjugate automata} 
Let $\A,\B$ be two automata.
A \emph{symbolic conjugacy}%
\index{symbolic!conjugacy of automata}%
\index{automaton!symbolic conjugate}
\index{conjugate automata!symbolic}
from $\A$ onto $\B$ is a pair $(\varphi,\psi)$ of
conjugacies $\varphi:X_\A\rightarrow X_\B$ and
$\psi:L_\A\rightarrow L_\B$ such that the following diagram is
commutative.
\begin{figure}[hbt]
\centering
\gasset{Nframe=n}
\begin{picture}(20,20)
\node(Xa)(0,20){$X_\A$}\node(Xb)(20,20){$X_\B$}
\drawedge(Xa,Xb){$\varphi$}
\node(La)(0,0){$L_\A$}\node(Lb)(20,0){$L_\B$}
\drawedge(Xa,La){$\lambda_\A$}\drawedge(Xb,Lb){$\lambda_\B$}
\drawedge(La,Lb){$\psi$}
\end{picture}
\end{figure}

\subsection{Splitting and merging maps}

Let $A,B$ be two alphabets and let $f:A\rightarrow B$ be a map from
$A$ onto $B$.  Let $X$ be a shift space on the alphabet $A$.  We
consider the set of words $A'=\{f(a_1)a_2\mid a_1a_2\in\B_2(X)\}$ as a
new alphabet.  Let $g:\B_2(X)\rightarrow A'$ be the $2$-block
substitution defined by $g(a_1a_2)=f(a_1)a_2$.

The \emph{in-splitting map}%
\index{in-splitting map}\index{map!in-splitting} defined on $X$ and
relative to $f$ or to $g$ is the sliding block map $g_\infty^{1,0}$
corresponding to $g$. It is a conjugacy from $X$ onto its image by
$X'=g_\infty^{1,0}(X)$ since its inverse is $1$-block.  The shift
space $X'$, is called the \emph{in-splitting}%
\index{shift space!in-splitting} of $X$, relative to $f$ or $g$.  The
inverse of an in-splitting map is called an \emph{in-merging map}\index{map!in-merging}.

In addition, any renaming of the alphabet of a shift space is also
considered to be an in-splitting map (and an in-merging map).

\begin{example}
  Let $A=B$ and let $f$ be the identity on $A$. The out-splitting of a shift
  $X$ relative to $f$ is the second
  higher block shift of $X$.
\end{example}

The following proposition relates splitting maps to edge splittings as
defined in Section~\ref{subSectionConjugacy}.

\begin{proposition}\label{propSymbolicSplit}
  An in-splitting map on an edge shift  is an edge in-splitting
  map, and conversely.
\end{proposition}

\begin{proof}
  Let first $G=(Q,\E)$ be a graph, and let $f:\E\rightarrow I$ be a
  map from $\E$ onto a set $I$. Set $\E'=\{f(e_1)e_2\mid
  e_1e_2\in\B_2(X_G)\}$. Let $g:\B_2(X_G)\rightarrow \E'$ be the
  $2$-block substitution defined by $g(e_1e_2)=f(e_1)e_2$.  Let
  $G'=(Q',\E')$ be the graph on the set of states $Q'=I\times Q$ defined
  for $e'=f(e_1)e_2$ by $i(e')=(f(e_1),i(e_2))$ and
  $t(e')=(f(e_2),t(e_2))$. Define $h:\E'\rightarrow \E$ and
  $k:Q'\rightarrow Q$ by $h(f(e_1)e_2)=e_2$ for $e_1e_2\in \B_2(X_G)$
  and $k(i,q)=q$ for $(i,q)\in I\times Q$. Then the pair $(h,k)$ is an
  in-merge from $G'$ onto $G$ and $h_\infty$ is the inverse of
  $g_\infty^{1,0}$. Indeed, one may verify that $(h,k)$ is a graph
morphism from $G'$ onto $G$. Next it is an in-merge because for
each $p,q\in Q$, the partition $(\E_p^q(t))_{t\in k^{-1}(q)}$ of
  $\E_p^q$
is defined by $\E_p^q(i,q)=E_p^q\cap f^{-1}(i)$.
  
  Conversely, set $G=(Q,\E)$ and $G'=(Q',\E')$. Let $(h,k)$ be an
  in-merge from $G'$ onto $G$. Consider the map $f:\E\rightarrow Q'$
  defined by $f(e)=r$ if $r$ is the common end of the edges in
  $h^{-1}(e)$.  The map $\alpha$ from $\E'$ to $Q'\times \E$ defined by
  $\alpha(i)=(r,h(i))$ where $r$ is the origin of $i$ is a bijection
  by definition of an in-merge.
  
  Let us show that, up to the bijection $\alpha$, the in-splitting map
  relative to $f$ is inverse of the map $h_\infty$. For $e_1,e_2\in
  \E$, let $r=f(e_1)$ and $e'=\alpha^{-1}(r,e_2)$. Then $h(e')=e_2$
  and thus $h_\infty$ is the inverse of the map $g_\infty^{1,0}$
  corresponding to the $2$-block substitution $g(e_1e_2)=(r,e_2)$.

\end{proof}

Symmetrically an \emph{out-splitting map}
\index{out-splitting map}
\index{map!out-splitting}
 is defined by the substitution
$g(ab)=af(b)$. Its inverse is an out-merging map.

We use the term  splitting to mean either a 
in-splitting or out-splitting. The same convention holds for a merging.

The following result, from~\cite{Nasu:1986}, is a generalization of
the Decomposition 
Theorem (Theorem~\ref{ShiftDecompositionTheorem}) to
arbitrary shift spaces.

\begin{theorem}\label{SymbolicDecompositionTheorem}
  Any conjugacy between shift spaces is a composition of splitting and
  merging maps.
\end{theorem}
The proof is similar to the proof of Theorem
~\ref{ShiftDecompositionTheorem}. It relies on the following lemma,
similar to Lemma~\ref{lemmaDecomp}.

\begin{lemma}\label{lemmaDecompSymbol}
  Let $\varphi:X\rightarrow Y$ be a $1$-block conjugacy whose inverse
  has memory $m\ge 1$ and anticipation $n\ge 0$. There are
  in-splitting maps from $X,Y$ to $\tilde{X},\tilde{Y}$ respectively
  such that the $1$-block conjugacy $\tilde{\varphi}$ making the diagram
  below commutative has an inverse with memory $m-1$ and anticipation $n$.
\begin{figure}[hbt]
\centering
\gasset{Nframe=n}
\begin{picture}(20,20)
\node(G)(0,20){$X$}\node(tG)(20,20){$\tilde{X}$}
\node(H)(0,0){$Y$}\node(tH)(20,0){$\tilde{Y}$}
\drawedge(G,tG){}\drawedge(H,tH){}
\drawedge(G,H){$\varphi$}\drawedge(tG,tH){$\tilde{\varphi}$}
\end{picture}
\end{figure}
\end{lemma}

\begin{proof}
  Let $A,B$ the alphabets of $X$ and $Y$ respectively. Let
  $h:A\rightarrow B$ be the $1$-block substitution such that
  $\varphi=h_\infty$. Let $\tilde{X}$ be the in-splitting of $X$
  relative to the map $h$. Set $A'=\{h(a_1)a_2\mid a_1a_2\in
  \B_2(X)\}$.  Let $\tilde{Y}=Y^{[2]}$ be the second higher block
  shift of $Y$ and let $B'=\B_2(Y)$.  Let $\tilde{h}:A'\rightarrow B'$
  be the $1$-block substitution defined by
  $\tilde{h}(h(a_1)a_2)=h(a_1)h(a_2)$. Then the $1$-block map
  $\tilde{\varphi}=\tilde{h}_\infty$ has the required properties.
\end{proof}
Lemma~\ref{lemmaDecompSymbol} has a dual where $\varphi$ is a
$1$-block map whose inverse has memory $m\ge 0$ and anticipation $n\ge
1$ and where in-splits are replaced by out-splits.

\begin{proof}[Proof of Theorem~\ref{SymbolicDecompositionTheorem}]
  Let $\varphi:X\rightarrow Y$ be a conjugacy from $X$ onto $Y$.
  Replacing $X$ by a higher block shift, we may assume that $\varphi$
  is a $1$-block map. Using iteratively Lemma~\ref{lemmaDecompSymbol},
  we can replace $\varphi$ by a $1$-block map whose inverse has memory
  0. Using then iteratively the dual of Lemma~\ref{lemmaDecompSymbol},
  we finally obtain a $1$-block map whose inverse is also $1$-block
  and is thus just a renaming of the symbols.
\end{proof}

\intertitre{Symbolic strong shift equivalence}

Let $M$ and $M'$ be two alphabetic $Q\times Q$-matrices over the
alphabets $A$ and $B$, respectively.  We say
that $M$ and $M'$ are
\emph{similar}\index{similar matrices}\index{matrix!similar} if they
are equal up to a bijection of $A$ onto $B$. We write
$M\leftrightarrow M'$ when $M$ and $M'$ are similar.  We say that two
alphabetic square
matrices $M$ and $M'$ over the alphabets $A$ and $B$ respectively are
\emph{symbolic elementary equivalent}%
\index{equivalent matrices!symbolic elementary}%
\index{matrix!symbolic elementary equivalent}%
\index{symbolic!elementary equivalent matrices} if there exist two
alphabetic 
matrices $R,S$ over the alphabets $C$ and $D$ respectively such that
\begin{displaymath}
  M\leftrightarrow RS,\quad M'\leftrightarrow SR\,.
\end{displaymath}
In this definition, the sets $CD$ and $DC$ of two letter words are
identified with  alphabets in bijection with $A$ and $B$, respectively.

We say that two matrices $M,M'$
are \emph{symbolic strong shift equivalent}%
\index{symbolic!strong shift equivalent matrices}
\index{strong shift equivalent matrices}\index{matrix!symbolic strong
  shift equivalent} \index{equivalent matrices!symbolic strong shift}
if there is a sequence $(M_0,M_1,\ldots,M_n)$ of alphabetic matrices such that $M_i$ and
$M_{i+1}$ are symbolic elementary equivalent for $0\le i<n$ with
$M_0=M$ and $M_n=M'$.

We introduce the following notion.  An automaton
$\A$ on the alphabet $A$
is said to be \emph{bipartite}
\index{bipartite automaton}
\index{automaton!bipartite}
 if there are partitions $Q=Q_1\cup Q_2$
of the set of states and $A=A_1\cup A_2$ of the alphabet such that all
edges labeled in $A_1$ go from $Q_1$ to $Q_2$ and all edges  labeled
in $A_2$ go from $Q_2$ to $Q_1$.

Let $\A$ be a bipartite automaton. Its adjacency matrix has the form
\begin{displaymath}
M(\A)=\begin{bmatrix}0&M_1\\M_2&0\end{bmatrix}
\end{displaymath}
where $M_1$ is a $Q_1\times Q_2$-matrix with elements in $\N\langle
A_1\rangle$
and $M_2$ is a $Q_2\times Q_1$-matrix with elements in $\N\langle
A_2\rangle$
The automata $\A_1$ and $\A_2$ which have $M_1M_2$ and $M_2M_1$ respectively
as adjacency matrix are called the \emph{components}
\index{components}\index{bipartite automaton!components} of $\A$
and the pair $\A_1,\A_2$ is a 
\emph{decomposition}\index{decomposition}\index{automaton!decomposition}
 of $\A$. We denote $\A=(\A_1,\A_2)$ a bipartite automaton $\A$ with
components $\A_1,\A_2$.
Note that $\A_1,\A_2$ are automata on the alphabets $A_1A_2$ and
$A_2A_1$ respectively.
\begin{proposition}\label{propBipartiteAutomata}
  Let $\A=(Q,E)$ be a bipartite  deterministic essential
  automaton.  Its components $\A_1,\A_2$ are 
  deterministic essential automata which are symbolic conjugate. If
  moreover $\A$ is
  strongly
connected (resp. reduced, resp. synchronized), then $\A_1,\A_2$ are
strongly connected (resp.reduced, resp. synchronized).
\end{proposition}
\begin{proof}
Let $Q=Q_1\cup Q_2$ and $A=A_1\cup A_2$ be the partitions
of the set $Q$ and the alphabet $A$ corresponding to the decomposition
$\A=(\A_1,\A_2)$.
It is clear that $\A_1,\A_2$ are  deterministic and that they are
strongly
connected if $\A$ is strongly connected.

Let $\varphi:X_{\A_1}\rightarrow X_{\A_2}$ be the conjugacy defined as follows.
For any $y=(y_n)_{n\in \Z}$ in $X_{\A_1}$ there is an $x=(x_n)_{n\in
  \Z}$ in $X_\A$ such that $y_n=x_{2n}x_{2n+1}$. Then $z=(z_n)_{n\in
  \Z}$ with
$z_n=x_{2n+1}x_{2n}$ is an element of $X_{\A_2}$. We define
$\varphi(y)=z$.
The analogous map $\psi:L_{\A_1}\rightarrow L_{\A_2}$ is such that
$(\varphi,\psi)$ is a symbolic conjugacy from $\A_1$ onto $\A_2$.

Assume that $\A$ is reduced. For $p,q\in Q_1$, there is a word
$w$ such that $w\in L_p(\A)$ and $w\notin L_q(\A)$ (or conversely). Set
$w=a_1a_2\cdots a_n$ with $a_i\in A$. If $n$ is even, then $(a_1a_2)\cdots(a_{n-1}a_n)$
is in $L_p(\A_1)$ but not in $L_q(\A_1)$. Otherwise, since $\A$
is essential, there is a letter $a_{n+1}$ such that
$wa_{n+1}$
is in $L_p(\A)$. Then $(a_1a_2)\cdots(a_na_{n+1})$ is in $L_p(\A_1)$
but not in $L_q(\A_1)$. Thus $\A_1$ is reduced. One proves in the same
way
that $\A_2$ is reduced.

Suppose finally that $\A$ is synchronized. Let $x$ be a synchronizing
word and set $x=a_1a_2\cdots a_n$ with $a_i\in A$. Suppose that all
paths
labeled $x$ end in $q\in Q_1$. Let $a_{n+1}$ be a letter such that
$q\cdot a_{n+1}\ne\emptyset$ and let $a_0$ be a letter such that
$a_0x$
is the label of at least one path. If $n$ is even,
then
$(a_1a_2)\cdots(a_{n-1}a_n)$ is synchronizing for $\A_1$ and 
$(a_0a_1)\cdots(a_na_{n+1})$ is synchronizing for $\A_2$. Otherwise,
$(a_0a_1)\cdots(a_{n-1}a_n)$ is synchronizing for $\A_1$ and 
$(a_1a_2)\cdots(a_na_{n+1})$ is synchronizing for $\A_2$.
\end{proof}
\begin{proposition}\label{propElemEquivIsSim}
Let $\A,\B$ be two automata such that $M(\A)$ and $M(\B)$ are
symbolic elementary equivalent. Then there is a bipartite automaton $\C=(\C_1,\C_2)$
such  that $M(\C_1),M(\C_2)$ are similar to $M(\A),M(\B)$ respectively.
\end{proposition}
\begin{proof}
Let $R,S$ be alphabetic matrices over alphabets
$C$ and $D$ respectively such that $M(\A)\leftrightarrow RS$ and
$M(\B)\leftrightarrow SR$. Let $\C$ be the bipartite automaton
on the alphabet
$C\cup D$ which is defined by the adjacency matrix 
\begin{displaymath}
M({\cal C})=\begin{bmatrix}0&R\\S&0\end{bmatrix}
\end{displaymath}
Then $M(\A)$ is similar to $M(\C_1)$ and $M(\B)$ is similar to
$M(\C_2)$.
\end{proof}
\begin{proposition}\label{BipartiteProp}
If the adjacency matrices of two automata are symbolic strong shift
equivalent, the automata are symbolic conjugate.
\end{proposition}
\begin{proof}
Since a composition of conjugacies is a conjugacy, it is enough to
consider the case where the adjacency matrices are symbolic elementary
equivalent.
Let $\A,\B$ be such that $M(\A),M(\B)$ are symbolic elementary equivalent.
By Proposition~\ref{propElemEquivIsSim}, there is a bipartite
automaton $\C=(\C_1,\C_2)$ such that $M(\C_1),M(\C_2)$ are similar to $M(\A)$
and $M(\B)$ respectively. By Proposition~\ref{propBipartiteAutomata}, the
automata $\C_1,\C_2$ are symbolic conjugate.
Since automata with similar adjacency matrices are obviously symbolic
conjugate,
the result follows.
\end{proof}
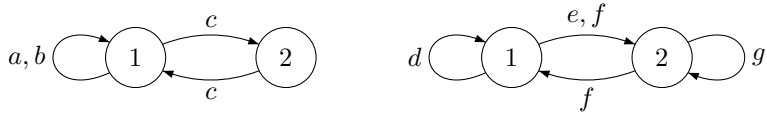
\begin{figure}[hbt]
\centering
\begin{picture}(80,12)(0,-5)
\put(0,0){\begin{picture}(30,30)
\node(1)(0,0){$1$}\node(2)(20,0){$2$}
\drawloop[loopangle=180](1){$a,b$}
\drawedge[curvedepth=3](1,2){$c$}
\drawedge[curvedepth=3](2,1){$c$}
\end{picture}
}
\put(50,0){\begin{picture}(30,30)
\node(1)(0,0){$1$}\node(2)(20,0){$2$}
\drawloop[loopangle=180](1){$d$}
\drawedge[curvedepth=3](1,2){$e,f$}
\drawedge[curvedepth=3](2,1){$f$}
\drawloop[loopangle=0](2){$g$}
\end{picture}
}
\end{picture}
\caption{Two symbolic conjugate automata.}\label{figConjugateSofic}
\end{figure}
\begin{example}\label{ExConjugateSofic}
  Let $\A,\B$ be the automata represented on
  Figure~\ref{figConjugateSofic}.  The matrices $M(\A)$ and
$M(\B)$ are symbolic elementary equivalent.
  Indeed, we have $M(\A)\leftrightarrow RS$ and $M(\B)\leftrightarrow
  SR$ for
\begin{displaymath}
  R=\begin{bmatrix}x&y\\0&x\end{bmatrix},\quad 
  S=\begin{bmatrix}z&t\\t&0\end{bmatrix}.
\end{displaymath}
Indeed, one has
\begin{displaymath}
  RS=\begin{bmatrix}xz+yt&xt\\xt&0\end{bmatrix},\quad 
  SR=\begin{bmatrix}zx&zy+tx\\tx&ty\end{bmatrix}.
\end{displaymath}
Thus the following tables give two bijections between the alphabets.
\begin{displaymath}
  \begin{array}{|c|c|c|}\hline a&b&c\\\hline
    xz&yt&xt\\ \hline 
  \end{array}\,,\quad
  \begin{array}{|c|c|c|c|}\hline d&e&f&g\\\hline
    zx&zy&tx&ty\\ \hline 
  \end{array}\,.
\end{displaymath}
\end{example}

The following result is due to Nasu~\cite{Nasu:1986}.
The equivalence between conditions (i) and (ii) is a version, for sofic shifts, of the
Classification Theorem 
(Theorem 7.2.12 in \cite{Lind&Marcus:1995}). The equivalence between
conditions
(i) and (iii) is due to Krieger~\cite{Krieger:1984}.
\begin{theorem}\label{TheoremNasu}
Let $X,X'$ be two  sofic shifts (resp. irreducible sofic shifts) and
let $\A,\A'$ be their Krieger
(resp. Fischer)
automata. The following conditions
are equivalent.
\begin{enumerate}
\item[(i)] $X,X'$ are conjugate.
\item[(ii)] The adjacency matrices
of $\A,\A'$ are symbolic strong shift equivalent.
\item[(iii)] $\A,\A'$ are symbolic conjugate.
\end{enumerate}
\end{theorem}

\begin{proof}
  We prove the result for irreducible shifts. The proof of the general
  case is in~\cite{Nasu:1986}.
  
  Assume that $X,X'$ are conjugate.  By the Decomposition Theorem
  (Theorem~\ref{SymbolicDecompositionTheorem}), it is enough to
  consider the case where $X'$ is an in-splitting of $X$.  Let
  $f:A\rightarrow B$ be a map and let $A'=\{f(a_1)a_2\mid
  a_1a_2\in\B_2(X)\}$ in such a way that $X'$ is the in-splitting of
  $X$ relative to $f$. Let $C=A\cup B$ and let $Z$ be the shift space
  composed of all biinfinite sequences $\cdots
  a_if(a_i)a_{i+1}f(a_{i+1})\cdots$ such that $\cdots
  a_ia_{i+1}\cdots$ is in $X$. Then $Z$ is an irreducible sofic shift.
Let $\A$ be the Fischer automaton of
  $Z$. Then $\A$ is bipartite and its components recognize, up to a
  bijection of the alphabets, $X$ and $X'$ respectively.  By
  Proposition~\ref{propBipartiteAutomata} the components are the
  Fischer automata of $X$ and $X'$ respectively.  Since the components
  of a bipartite automaton have symbolic elementary equivalent
  adjacency matrices, this proves that (i) implies (ii).

  That (ii) implies (iii) is Proposition~\ref{BipartiteProp}. Finally,
  (iii) implies (i) by definition of symbolic conjugacy.
\end{proof}

\subsection{Symbolic conjugate automata}

The following result is due to Hamachi and
Nasu~\cite{Hamachi&Nasu:1988}. It shows that, in
Theorem~\ref{TheoremNasu},
the equivalence between conditions (ii) and (iii) holds
for automata which are not reduced.
\begin{theorem}\label{TheoremHamachiNasu}
  Two essential automata are symbolic conjugate if and only if their
  adjacency matrices are symbolic strong shift equivalent.
\end{theorem}
The first element of the proof is a version of the Decomposition
Theorem for automata.

Let $\A,\A'$ be two automata. An
\emph{in-split}\index{automaton!in-split}
 from $\A$ onto $\A'$
is a symbolic conjugacy $(\varphi,\psi)$ such that
$\varphi:X_\A\rightarrow X_{\A'}$ and $\psi:L_\A\rightarrow L_{\A'}$
are in-splitting maps. A similar definition holds for out-splits.

\begin{theorem}\label{DecompTheorAutomata}
  Any symbolic conjugacy between automata is a composition of splits
  and merges.
\end{theorem}
The proof relies on the following variant of
Lemma~\ref{lemmaDecompSymbol}.
\begin{lemma}\label{lemmaDecompAutomata}
Let $\alpha,\beta$ be $1$-block maps and $\varphi,\psi$ be 
 $1$-block conjugacies such  such that the diagram below on the left is
commutative.

If the  inverses of $\varphi,\psi$ have memory $m\ge 1$ and
anticipation $n\ge 0$, there exist in-splits  
$\tilde{X},\tilde{Y},\tilde{Z},\tilde{T}$ of $X,Y,Z,T$ respectively
 and  $1$-block maps $\tilde{\alpha}:\tilde{X}\to\tilde{Z}$,
$\tilde{\beta}:\tilde{Y}\to\tilde{T}$ such that the $1$-block
conjugacies $\tilde{\varphi},\tilde{\psi}$ making the diagram below on
the right commutative have inverses with memory $m-1$ and anticipation $n$.

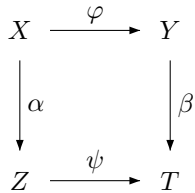
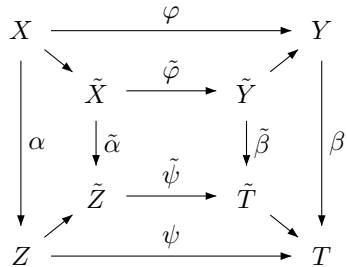
\begin{figure}[hbt]
\centering
\gasset{Nframe=n}
\begin{picture}(100,35)
\put(0,10){
\begin{picture}(20,20)
\node(X)(0,20){$X$}\node(Y)(20,20){$Y$}
\node(Z)(0,0){$Z$}\node(T)(20,0){$T$}
\drawedge(X,Y){$\varphi$}\drawedge(Z,T){$\psi$}
\drawedge(X,Z){$\alpha$}\drawedge(Y,T){$\beta$}
\end{picture}
}
\put(50,0){
\begin{picture}(40,30)
\node(X)(0,30){$X$}\node(Y)(40,30){$Y$}
\node(Z)(0,0){$Z$}\node(T)(40,0){$T$}
\node(tX)(10,22){$\tilde{X}$}\node(tY)(30,22){$\tilde{Y}$}
\node(tZ)(10,8){$\tilde{Z}$}\node(tT)(30,8){$\tilde{T}$}
\drawedge(X,tX){}\drawedge(tY,Y){}\drawedge(Z,tZ){}\drawedge(tT,T){}
\drawedge(X,Y){$\varphi$}\drawedge(Z,T){$\psi$}
\drawedge(X,Z){$\alpha$}\drawedge(Y,T){$\beta$}
\drawedge(tX,tY){$\tilde{\varphi}$}\drawedge(tZ,tT){$\tilde{\psi}$}
\drawedge(tX,tZ){$\tilde{\alpha}$}\drawedge(tY,tT){$\tilde{\beta}$}
\end{picture}
}
\end{picture}
\end{figure}
\end{lemma}
\begin{proof}
Let $A,B,C,D$ be the alphabets of $X,Y,Z$ and $T$ respectively.
Let $h:A\rightarrow B$ and $k:C\rightarrow D$ be the $1$-block
substitutions such that $\varphi=h_\infty$ and $\psi=k_\infty$.
Set $\tilde{A}=\{h(a_1)a_2\mid a_1a_2\in\B_2(X)\}$ and 
$\tilde{C}=\{k(c_1)c_2\mid c_1c_2\in\B_2(Z)\}$. Let
$\tilde{X}$ (resp. $\tilde{Z}$)
be the image of $X$ (resp. of $Z$) under the in-splitting map
relative to $h$ (resp. $k$). Set $\tilde{Y}=Y^{[2]}$, $\tilde{B}=\B_2(Y)$,
$\tilde{T}=T^{[2]}$ and $\tilde{D}=\B_2(T)$.
Define $\tilde{\alpha}$ and $\tilde{\beta}$ by 
\begin{displaymath}
\tilde{\alpha}(h(a_1)a_2)=k\alpha(a_1)\alpha(a_2),\quad
\tilde{\beta}(b_1b_2)=\beta(b_1)\beta(b_2)
\end{displaymath}
and $\tilde{h}:\tilde{A}\rightarrow \tilde{B}$, 
$\tilde{k}:\tilde{C}\rightarrow\tilde{D}$ by
\begin{displaymath}
\tilde{h}(h(a_1)a_2)=h(a_1)h(a_2),\quad \tilde{k}(k(c_1)c_2)=k(c_1)k(c_2)
\end{displaymath}
Then the $1$-block conjugacies $\tilde{\varphi}=\tilde{h}_\infty$ and
$\tilde{\psi}=\tilde{k}_\infty$ satisfy the conditions of the statement.
\end{proof}
\begin{proof}[Proof of Theorem~\ref{DecompTheorAutomata}]
Let $\A=(G,\lambda)$ and $\A'=(G',\lambda')$ be two automata
with $G=(Q,\E)$ and $G'=(Q',\E')$.
Let $(\varphi,\psi)$ be a symbolic conjugacy from $\A$ onto $\A'$.
Replacing $\A$ and $\B$ by some extension $\A^{[m,n]}$ and
$\B^{[m,n]}$
we may reduce to the case where $\varphi,\psi$ are $1$-block conjugacies. By
using repeatedly Lemma~\ref{lemmaDecompAutomata}, we may reduce to
the case where the inverses of $\varphi,\psi$ have memory $0$. Using
repeatedly the dual version of Lemma~\ref{lemmaDecompAutomata},
we are reduced to the case where $\varphi,\psi$ are renaming of the alphabets.
\end{proof}

The second step for the proof of Theorem~\ref{TheoremHamachiNasu}
is the following statement.

\begin{proposition}\label{propInSplitAutomaton}
Let $\A,\A'$ be two essential automata. If $\A'$ is an in-split of $\A$, the
matrices $M(\A)$ and $M(\A')$ are symbolic elementary equivalent.
\end{proposition}
\begin{proof}
Set $\A=(G,\lambda)$ and $\A'=(G',\lambda')$. Let $A'=\{f(a)b\mid
ab\in\B_2(L_\A)\}$ be the alphabet of $\A'$ for a map $f:A\rightarrow B$.
By Proposition~\ref{propSymbolicSplit}, the symbolic in-splitting  map from
$X_G$ onto $X_{G'}$ is also an in-splitting map. Thus there is an
in-merge $(h,k)$ from $G'$ onto $G$ such that the in-split
from $\A$ onto $\A'$ has the form $(h_\infty^{-1},\psi)$.
We define an alphabetic $Q'\times Q$-matrix $R$ and a $Q\times
Q'$-matrix
$S$ as follows. Let $r,t\in Q'$ and let $p=k(r)$, $q=k(t)$. 
Let $e$ be an edge of $\A'$ ending in $r$, 
and set
 $a=\lambda(h(e))$). Then the label of any edge going out of $r$
is of the form $f(a)b$ for some $b\in A$.
Thus $f(a)$ does not depend on $e$ but only on
$r$.
We define a map $\pi:Q'\rightarrow B$ by $\pi(r)=f(a)$.  Then, we set
\begin{displaymath}
R_{rp}=\begin{cases}\pi(r)&\text{if $k(r)=p$}\\0&\text{otherwise}\end{cases}
,\quad S_{pt}=M(\A)_{pq}
\end{displaymath}
Let us verify that  $M(\A')= RS$ and $M(\A)\leftrightarrow SR$. 
We first have for $r,t\in Q'$
\begin{displaymath}
(RS)_{rt}=\sum_{p\in Q}R_{rp}S_{pt}=\pi(r)M_{k(r)k(q)}=M(\A')_{rt}
\end{displaymath}
and thus $RS=M(\A')$.
Next, for $p,q\in Q$
\begin{displaymath}
(SR)_{pq}=\sum_{p\in Q}R_{rp}S_{pt}=\sum_{t\in
    k^{-1}(q)}M(\A)_{pq}\pi(t)
=\sum_{a\in A}(M(\A)_{pq},a)af(a)
\end{displaymath}
and thus $SR\leftrightarrow M(\A)$ using the bijection $a\rightarrow
af(a)$
between $A$ and $AB$.
\end{proof}

\begin{proof}[Proof of Theorem~\ref{TheoremHamachiNasu}]
The condition is sufficient by Proposition~\ref{BipartiteProp}.
Conversely, let $\A,\A'$ be two symbolic conjugate essential
automata. By Theorem~\ref{DecompTheorAutomata}, we may assume that
$\A'$ is a split of $\A$. We assume that $\A'$ is an in-split of $\A$.
By Proposition~\ref{propInSplitAutomaton}, the adjacency matrices of
$\A$ and $\A'$ are symbolic elementary equivalent.
\end{proof}

%
%
\section{Special families of automata}\label{sectionSpecialFamilies}

In this section, we consider two particular families of automata:
local automata and automata with finite delay. Local automata are
closely related to shifts of finite type. The main result is an
embedding theorem (Theorem~\ref{BBEP:TheoremLocal}) related to Nasu's
Masking Lemma (Proposition~\ref{BBEP:MaskingLemma}). Automata with
finite left and right delay are related to a class of shifts called
shifts of almost finite type (Proposition~\ref{BBP:Almostfinitetype}).

\subsection{Local automata}

Let $m,n\ge0$.  An automaton $\A=(Q,E)$ is said to be
$(m,n)$-\emph{local}%
\index{local automaton}\index{automaton!local} if whenever
$p\edge{u}q\edge{v}r$ and $p'\edge{u}q'\edge{v}r'$ are two paths with
$|u|=m$ and $|v|=n$, then $q=q'$. It is \emph{local} if it is
$(m,n)$-local for some $m,n$.

\begin{example}\label{BBEP:ExampleLocalAutomaton}
  The automaton represented in Figure~\ref{BBEP:figureLocal} is
  $(3,0)$-local.  Indeed, a simple inspection shows that each of the
  six words of length~$3$ which are labels of paths  uniquely determines
  its terminal vertex. It is also $(0,3)$-local. It is not $(2,0)$-local
  (check the word $ab$), but it is $(2,1)$-local and also
  $(1,2)$-local.

\begin{figure}[hbt]
  \centering
  \gasset{Nadjust=wh}
  \begin{picture}(20,20)
    \node(1)(10,0){$1$}\node(2)(0,15){$2$}\node(3)(20,15){$3$}
    \drawedge[curvedepth=3](1,2){$a,b$}\drawedge[curvedepth=3](2,3){$b$}
    \drawedge[curvedepth=3](3,1){$a$}
  \end{picture}
  \caption{A local automaton.}\label{BBEP:figureLocal}
\end{figure}
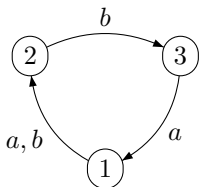
\end{example}

We say that an automaton $\A=(Q,E)$ is
\emph{contained}\index{automaton!contained in another} in an automaton
$\A'=(Q',E')$ if $Q\subset Q'$ and $E\subset E'$.
We note that if $\A$ is contained in $\A'$ and if $\A'$ is local, then
$\A$
is local.

\begin{proposition}\label{BBEP:propLocalAutomatonConjugacy}
  An essential automaton $\A$ is local if and only if the map
  $\lambda_\A:X_\A\to L_\A$ is a conjugacy from $X_\A$ onto $L_\A$.
\end{proposition}

\begin{proof}
  Suppose first that $\A$ is $(m,n)$-local. Consider an $m+1+n$-block
  $w= uav$ of $L_\A$, with $|u|=m$, $|v|=n$. All finite paths of $\A$
  labeled $w$ have the form $r\edge{u} p\edge{a}q\edge{v}s$ and share
  the same edge $p\edge{a}q$. This shows that $\lambda_\A$ is
  injective and that $\lambda_\A^{-1}$ is a map with memory $m$ and
  anticipation $n$.
  
  Conversely, assume that $\lambda_\A^{-1}$ exists, and that it has
  memory $m$ and anticipation $n$. We show that $\A$ is
  $(m+1,n)$-local. Let 
  \begin{displaymath}
    r\edge{u} p\edge{a}q\edge{v}s\quad\text{and}\quad 
    r'\edge{u}p'\edge{a}q'\edge{v}s'
  \end{displaymath}
  and be two paths of length $m+1+n$, with $|u|=m$, $|v|=n$ and $a$ a
  letter. Since $\A$ is essential, there exist two biinfinite paths
  which contain these finite paths, respectively.  Since
  $\lambda_\A^{-1}$ has memory $m$ and anticipation $n$, the blocks
  $uav$ of the biinfinite words carried by these paths are mapped by
  $\lambda_\A^{-1}$ onto the edges $p\edge{a}q$ and $p'\edge{a}q'$
  respectively. This shows that $p=p'$ and $q=q'$.
\end{proof}
  

The next statement is Proposition~10.3.10
in~\cite{Berstel&Perrin&Reutenauer:2009}. 
\begin{proposition}
  The following conditions are equivalent for a strongly connected
  finite automaton $\A$.
\begin{conditionsiii}
\item $\A$ is local;
\item distinct cycles have distinct labels.
\end{conditionsiii}
\end{proposition}
 Two cycles in this statement are considered to be
distinct if, viewed as paths, they are distinct. 

The following result shows the strong connection between
shifts of finite type and local automata. It gives an effective method
to verify whether or not a shift space is  of finite type.

\begin{proposition}\label{BBEP:propLocalisSFT}
A shift space (resp. an irreducible shift space) is of finite type
if and only if its Krieger automaton (resp. its Fischer automaton) is local.
\end{proposition}
\begin{proof}
\noindent  Let $X=X^{(W)}$ for a finite set $W\subset A^*$. We may assume that
all words of $W$ have the same length $n$. Let $\A=(Q,i,Q)$ be the
$(n,0)$-local deterministic automaton defined as follows. The set of states
is $Q=A^n\setminus W$ and there is an edge $(u,a,v)$ for every
$u,v\in Q$ and $a\in A$ such that $ua\in Av$. Then   $\A$
recognizes the set $\B(X)$. Since the reduction of
a local automaton is local, the minimal automaton of $\B(X)$ is local.
Since the Krieger automaton of $X$ is contained in the minimal
automaton
of $\B(X)$, it is local. If $X$ is irreducible, then its Fischer
automaton
is also local since it is contained in the Krieger automaton.

Conversely, Proposition~\ref{BBEP:propLocalAutomatonConjugacy} implies
that a shift space recognized by a local automaton is conjugate to a
shift of finite type and thus is of finite type.
\end{proof}

\begin{example}
Let $X$ be the shift of finite type on the alphabet $A=\{a,b\}$
defined by the forbidden factor $ba$. The Krieger automaton of $X$
is represented on Figure~\ref{figReducibleSFT}. It is $(1,0)$-local.
\begin{figure}[htbp]
    \centering
\gasset{Nadjust=wh}
\begin{picture}(40,13)(0,-2)
\node(1)(0,0){$1$}
\node(2)(20,0){$2$}

\drawloop[loopangle=90](1){$a$}
\drawloop[loopangle=90](2){$b$}
\drawedge(1,2){$b$}
\end{picture}
 
\caption{The Krieger automaton of a reducible shift of finite type. }\label{figReducibleSFT}
\end{figure}
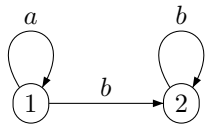
\end{example}

For $m,n\ge 0$, the \emph{standard}\index{standard local automaton}%
\index{local automaton!standard}\index{automaton!standard local}
$(m,n)$-local automaton is the automaton with states the set of words
of length $m+n$ and edges the triples $(uv,a,u'v')$ for $u,u'\in A^m$,
$a\in A$ and $v,v'\in A^n$ such that for some letters $b,c\in A$, one
has $uvc=bu'v'$ and $a$ is the first letter of $vc$.

The standard $(m,0)$-local automaton is also called the De Bruijn
automaton of order $m$.
\begin{example}
The standard $(1,1)$-local automaton on the alphabet $\{a,b\}$ is represented on Figure~\ref{figFree}.
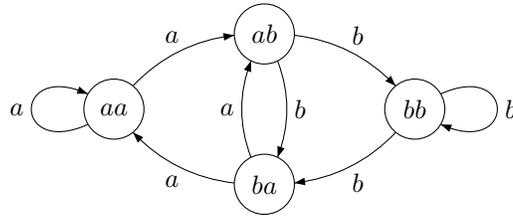
\begin{figure}[htbp]
    \centering
\begin{picture}(40,22)(0,-2)
\node(1)(0,10){$aa$}
\node(2)(20,20){$ab$}
\node(3)(40,10){$bb$}
\node(4)(20,0){$ba$}
\drawloop[loopangle=180](1){$a$}
\drawloop[loopangle=0](3){$b$}
\drawedge[curvedepth=3](1,2){$a$}
\drawedge[curvedepth=3](2,3){$b$}
\drawedge[curvedepth=3](3,4){$b$}
\drawedge[curvedepth=3](4,1){$a$}
\drawedge[curvedepth=3](4,2){$a$}
\drawedge[curvedepth=3](2,4){$b$}
\end{picture}
 
\caption{The standard $(1,1)$-local automaton. }\label{figFree}
\end{figure}
\end{example}

\intertitre{Complete automata}

An automaton $\A$ on the alphabet $A$ is called
\emph{complete}\index{complete automaton}\index{automaton!complete} if
any word on $A$ is the label of some path in $\A$. As an example, the
standard $(m,n)$-local automaton is complete.

The following result is from~\cite{Beal&Lombardy&Perrin:2008}.

\begin{theorem}\label{BBEP:TheoremLocal}
  Any local automaton is contained in a complete local automaton.
\end{theorem}

\noindent The proof relies on the following version of the masking
lemma.

\begin{proposition}[Masking lemma]\label{BBEP:MaskingLemma}
  Let $\A$ and $\B$ be two automata and assume that $M(\A)$ and
  $M(\B)$ are elementary equivalent. If $\B$ is contained in an
  automaton $\B'$, then $\A$ is contained in some automaton $\A'$
  which is conjugate to $\B'$.
\end{proposition}
\begin{proof}
  Let $\A=(Q,E)$, $\B=(R,F)$ and $\B'=(R',F')$. Let $D$ be an $R\times
  Q$ nonnegative integral matrix and $N$ be an alphabetic $Q\times R$
  matrix such that $M(\A)=ND$ and $M(\B)=DN$.  Set $Q'=Q\cup
  (F'\setminus F)$. Let $D'$ be the $R'\times Q'$ nonnegative integral
  matrix defined for $r\in R'$ and $u\in Q'$ by
\begin{displaymath}
D'_{ru}=\begin{cases}
D_{ru}&\mbox{if }
r\in R,u\in Q\\
1&\text{if $u\in F'\setminus F$ and $u$ starts in $r$}\\
0&\mbox{otherwise}
\end{cases}
\end{displaymath}
Let $N'$ be the alphabetic $Q'\times R'$ matrix defined for $a\in A$
for $u\in Q'$ and $s\in R'$ by
\begin{displaymath}
(N'_{us},a)=
\begin{cases}
  (N_{us},a)&\mbox{if }u\in Q,s\in R\\
  1&\text{if $u\in F'\setminus F$ and $u$ is labeled with $a$ and ends
  in $s$},\\
  0&\mbox{otherwise. }
\end{cases}
\end{displaymath}
Then $N'D'$ is the adjacency matrix of an automaton $\A'$. By
definition,
$\A'$ contains $\A$ and it is conjugate to $\B'$ by
Proposition~\ref{BBEP:PropInSplit}.
\end{proof}

We illustrate the proof of Proposition~\ref{BBEP:MaskingLemma} by the
following example.

\begin{example}
  Consider the automata $\A$ and $\B$ given in
  Figure~\ref{BBEP:FigureMasking}. The automaton $\A$ is the local
  automaton of Example~\ref{BBEP:ExampleLocalAutomaton}. The automaton $\B$
  is an in-split of $\A$. Indeed, we have $M(\A)=ND$, $M(\B)=DN$ with
\begin{displaymath}
  N=\begin{bmatrix}0&a&b&0\\0&0&0&b\\a&0&0&0\end{bmatrix}\qquad
  D=\begin{bmatrix}1&0&0\\0&1&0\\0&1&0\\0&0&1\end{bmatrix}.
\end{displaymath}
\begin{figure}[hbt]
  \centering
  \gasset{Nadjust=wh}
  \begin{picture}(70,20)(0,-2)
    \put(0,0){
      \begin{picture}(20,30)        
        \node(1)(10,0){$1$}\node(2)(0,15){$2$}\node(3)(20,15){$3$}        
        \drawedge[curvedepth=3](1,2){$a,b$}\drawedge[curvedepth=3](2,3){$b$}
        \drawedge[curvedepth=3](3,1){$a$}
      \end{picture}      
    }
    \put(50,0){
      \begin{picture}(20,20)       
        \node(1)(10,0){$1$}\node(2)(0,15){$2$}\node(3)(10,10){$3$}
        \node(4)(20,15){$4$}       
        \drawedge[curvedepth=3](1,2){$a$}\drawedge(1,3){$b$}
        \drawedge[curvedepth=3](2,4){$b$}\drawedge(3,4){$b$}
        \drawedge[curvedepth=3](4,1){$a$}
      \end{picture}
    }
  \end{picture}
  \caption{The automaton $\B$ on the right is an in-split of the local
    automaton $\A$ on the left.}\label{BBEP:FigureMasking}
\end{figure}
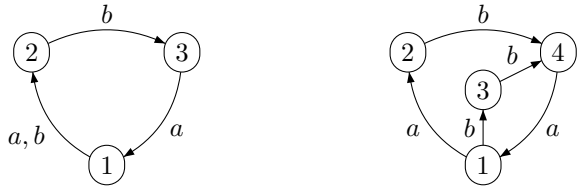

We have represented on the right of Figure~\ref{BBEP:FigureMaskingBis}
the completion of $\B$ as a complete local automaton with the same
number of states. On the left, the construction of the proof of
Proposition~\ref{BBEP:MaskingLemma} has been carried on to produce a
local automaton containing $\A$.
\begin{figure}[hbt]
  \centering
  \gasset{Nw=5,Nh=5}
  \begin{picture}(70,38)(0,-5)
    \put(-10,0){
      \begin{picture}(20,35) (-10,0)       
        \node(1)(10,0){$1$}\node(2)(0,15){$2$}\node(3)(20,15){$3$}
        \node[AHangle=30,linewidth=.5,linecolor=blue](e)(-10,5){$e$}
        \node[AHangle=30,linewidth=.5,linecolor=blue](f)(10,25){$f$}
        \node[AHangle=30,linewidth=.5,linecolor=blue](g)(30,0){$g$}
        \drawedge[curvedepth=3](1,2){$a,b$}\drawedge[curvedepth=3,ELside=r](2,3){$b$}
        \drawedge[curvedepth=3,ELside=r](3,1){$a$}
        \gasset{AHangle=30,linewidth=.5,linecolor=blue}
        \drawedge[curvedepth=3](1,e){$a$}\drawloop[loopangle=180](e){$a$}
        \drawedge(e,2){$a$}
        \drawedge[curvedepth=3](2,f){$b$}\drawloop[loopangle=90,ELside=r](f){$b$}
        \drawedge[curvedepth=2](f,3){$b$}
        \drawedge[curvedepth=2](1,g){$b$}\drawedge[curvedepth=2](g,1){$a$}
      \end{picture}
    }
    \put(50,0){
      \begin{picture}(20,20)        
        \node(1)(10,0){$1$}\node(2)(0,15){$2$}\node(3)(10,10){$3$}\node(4)(20,15){$4$}        
        \drawedge[curvedepth=3](1,2){$a$}\drawedge[curvedepth=3](1,3){$b$}
        \drawedge[curvedepth=3](2,4){$b$}\drawedge(3,4){$b$}
        \drawedge[curvedepth=3](4,1){$a$}
        \gasset{AHangle=30,linewidth=.5,linecolor=blue}
        \drawloop[loopangle=0](4){$b$}\drawedge[curvedepth=2](3,1){$a$}\drawloop[loopangle=180](2){$a$}
        
      \end{picture}
    }
  \end{picture}
  \caption{The automata $\A'$ and $\B'$. Additional edges are drawn thick.}\label{BBEP:FigureMaskingBis}
\end{figure}
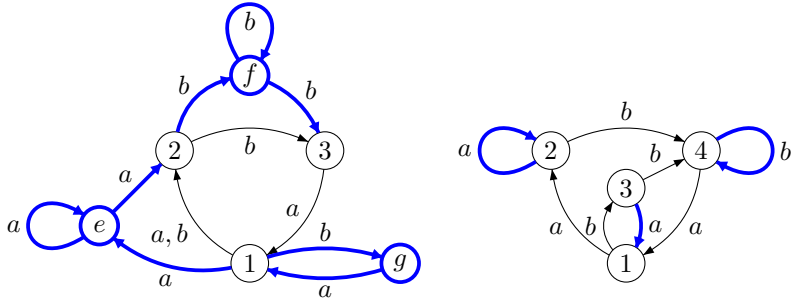
In terms of adjacency matrices, we have $ M(\A')=N'D'$, $M(\B')=D'N'$
with
\begin{displaymath}
 N'=\begin{bmatrix}0&a&b&0\\0&0&0&b\\a&0&0&0\\
\textcolor{red}{0}&\textcolor{red}{a}&\textcolor{red}{0}&\textcolor{red}{0}\\
\textcolor{red}{0}&\textcolor{red}{0}&\textcolor{red}{0}&\textcolor{red}{b}\\
\textcolor{red}{a}&\textcolor{red}{0}&\textcolor{red}{0}&\textcolor{red}{0}\end{bmatrix},
\qquad
D'=\begin{bmatrix}1&0&0&\textcolor{red}{0}&\textcolor{red}{0}&\textcolor{red}{0}\\
0&1&0&\textcolor{red}{1}&\textcolor{red}{0}&\textcolor{red}{0}\\
0&1&0&\textcolor{red}{0}&\textcolor{red}{0}&\textcolor{red}{1}\\
0&0&1&\textcolor{red}{0}&\textcolor{red}{1}&\textcolor{red}{0}
\end{bmatrix}
\end{displaymath}
\end{example}
\begin{proof}[Proof of Theorem~\ref{BBEP:TheoremLocal}]
  Since $\A$ is local, the map $\lambda_\A$ is a conjugacy from $X_\A$
  to $L_\A$. Let $(m,n)$ be the memory and anticipation of
  $\lambda_\A^{-1}$. There is a sequence $(\A_0,\ldots,\A_{m+n})$ of
  automata such that $\A_0=\A$, each $\A_i$ is a split or a merge of
  $\A_{i-1}$ and $\A_{n+m}$ is contained in the standard $(n+m)$-local
  automaton.  Applying iteratively
  Proposition~\ref{BBEP:MaskingLemma}, we obtain that $\A$ is
  contained in an automaton which is conjugate to the
  standard $(m,n)$-local automaton and which is thus complete.
\end{proof}

%
%

\subsection{Automata with finite delay}

An automaton is said to have \emph{right delay}%
\index{right!delay}\index{delay!right}\index{automaton!right delay}
$d\ge 0$ if for any pair of paths
\begin{displaymath}
p\edge{a}q\edge{z}r,\quad p\edge{a}q'\edge{z}r'
\end{displaymath}
with $a\in A$, if $|z|=d$, then $q=q'$.  Thus a deterministic
automaton has right delay $0$. An automaton has \emph{finite right delay} if
it has right delay $d$ for some (finite) integer $d$. Otherwise, it is
said to have \emph{infinite right delay}.

\begin{example}
The automaton represented on Figure~\ref{figRightDelay} has right delay $1$.
\end{example}

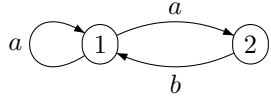
\begin{figure}[hbt]
\centering\gasset{Nadjust=wh}
\begin{picture}(20,10)(0,-5)
\node(1)(0,0){$1$}\node(2)(20,0){$2$}
\drawloop[loopangle=180](1){$a$}\drawedge[curvedepth=3](1,2){$a$}
\drawedge[curvedepth=3](2,1){$b$}
\end{picture}
\caption{A automaton with right delay $1$}\label{figRightDelay}
\end{figure}

\begin{proposition}\label{propInfDel}
  An automaton has infinite right delay if and only if there exist
  paths $p\edge{v}q\edge{u}q$ and $p\edge{v}q'\edge{u}q'$ with $q\ne
  q'$ and $|u|>0$.
\end{proposition}

The following statement is Proposition 5.1.11 in~\cite{Lind&Marcus:1995}.
\begin{proposition}\label{propFiniteRightDelay}
  An automaton has finite right delay if and only if it is conjugate
  to a deterministic automaton.
\end{proposition}

In the same way the automaton is said to have \emph{left delay}%
\index{left!delay}\index{delay!left}\index{automaton!left delay} $d\ge
0$ if for any pair of paths $p\edge{z}q\edge{a}r$ and
$p'\edge{z}q'\edge{a}r$ with $a\in A$, if $|z|=d$, then $q=q'$.

\begin{corollary}
  If two automata are conjugate, and if one has finite right
  (left) delay, then the other also has.
\end{corollary}

\begin{proposition}\label{propLocalHasFinite}
  An essential $(m,n)$-local automaton has right delay $n$ and left
  delay $m$.
\end{proposition}
\begin{proof}
  Let $p\edge{a}q\edge{z}r$ and $p\edge{a}q'\edge{z}r'$ be two paths
  with $a\in A$ and $|z|=n$. Since $\A$ is essential there is a path
  $u\edge{y}p$ of length $m$ in $\A$. Since $\A$ is $(m,n)$-local, we
  have $q=q'$.  Thus $\A$ has right delay $n$. The proof for the
  left delay $m$ is symmetrical.
\end{proof}

A shift space is said to have \emph{almost finite type}\index{almost
  finite type shift}\index{shift space!almost finite type} if it can
be recognized by a strongly connected automaton with both finite left
and finite right delay.

An irreducible shift of finite type is also of almost finite type
since a local automaton has finite right and left delay by
Proposition~\ref{propLocalHasFinite}.

\begin{example}
  The even shift has almost finite type. Indeed, the automaton of
  Figure \ref{figAutomatonEven} on the right has right and left delay~$0$.
\end{example}
The following result is from~\cite{Nasu1985}.
\begin{proposition}\label{BBP:Almostfinitetype}
An irreducible shift space is of almost finite type if and only if its
Fischer automaton has finite left delay.
\end{proposition}
\begin{proof}
  The condition is obviously sufficient. Conversely, let $X$ be a
  shift of almost finite type. Assume the Fischer automaton $\A=(Q,E)$
  of $X$ does not have finite left delay. Let, in view of
  Proposition~\ref{propInfDel} $u,v\in A^*$ and $p,q,q'\in Q$ with
  $q\ne q'$ be such that $q\cdot u=q$, $q'\cdot u=q'$ and $p=q\cdot
  v=q'\cdot v$.  Since $\A$ is strongly connected, there is a word $w$
  such that $p\cdot w=q$.

  Let $\B=(R,F)$ be an automaton with finite right and left delay
  which recognizes $X$. By Proposition~\ref{propFiniteRightDelay}, we
  may assume that $\B$ is deterministic. Let $\varphi:R\rightarrow Q$
  be a reduction from $\B$ onto $\A$. Since $R$ is finite, there is an
  $x\in u^+$ such that $r\cdot x=r\cdot x^2$ for all $r\in R$ (this
  means that the map $r\mapsto r\cdot x$ is idempotent; such a 
  word exists since each element in the finite transition semigroup of
  the automaton $\B$ has a power
  which is an idempotent). Set
\begin{displaymath}
  S=R\cdot x,\quad T=\varphi^{-1}(q)\cap S,
  \quad T'=\varphi^{-1}(q')\cap S
\end{displaymath}
Since $q\ne q'$, we have $T\cap T'=\emptyset$.
For any $t\in T$, we have $\varphi(t\cdot vw)=q$ and thus $t\cdot
vwx\in T$. For $t,t'\in T$ with $t\ne t'$, we cannot have $t\cdot
vwx=t'\cdot vwx$ since
otherwise $\B$ would have infinite left delay. Thus the map $t\mapsto
t\cdot vwx$ is a bijection of $T$. 

Let $t'\in T'$. Since $\varphi(t'\cdot vw)=q$, we have $t'\cdot vwx\in
T$.
Since the action of $vwx$ induces a permutation on $T$, there exists
$t\in T$ such that $t\cdot vwx=t'\cdot vwx$. This contradicts the fact
that
$\B$ has finite left delay.
\end{proof}
\begin{example}
The deterministic automaton represented on Figure~\ref{figInfiniteleftDelay} has infinite
left delay. Indeed, there are paths $\cdots 1\edge{b}1\edge{a}1$
and $\cdots 2\edge{b}2\edge{a}1$. Since this automaton cannot be reduced,
$X=L_\A$ is not of almost finite type.
\begin{figure}[hbt]
\centering\gasset{Nadjust=wh}
\begin{picture}(30,15)(0,-5)
\node(1)(0,0){$1$}\node(2)(20,0){$2$}
\drawloop[loopangle=90](1){$a,b$}\drawedge[curvedepth=5](1,2){$c$}
\drawedge[curvedepth=5](2,1){$a$}\drawloop[loopangle=0](2){$b$}
\end{picture}
\caption{An automaton with infinite left delay}\label{figInfiniteleftDelay}
\end{figure}
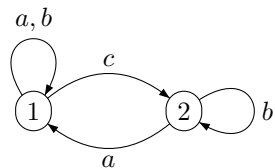
\end{example}
%
%
\section{Syntactic invariants}\label{sectionSyntacticInvariants}

We introduce in this section the syntactic graph of an automaton.  It
uses the Green relations in the transition semigroup of the automaton.
We show that the syntactic graph is an invariant for symbolic
conjugacy (Theorem~\ref{theoremSyntGraphAutomata}). The proof uses
bipartite automata.

The final subsection considers the characterization of sofic shifts
with respect to the families of ordered semigroups known as
pseudovarieties.

\subsection{The syntactic graph}

Let $\A=(Q,E)$ be a deterministic automaton on the alphabet
$A$. Each
word $w\in A^*$ defines a partial map denoted by $\varphi_\A(w)$ from $Q$ to $Q$ which maps
$p\in Q$ to $q\in Q$ if  $p\cdot w=q$. The transition
semigroup
of $\A$, already defined in
Section~\ref{subSectionsyntacticSemigroup}, 
is the image of $A^+$ by the morphism $\varphi_\A$ (in
this subsection, we will not use the order on the transition semigroup).

We give a short summary of \emph{Green relations}\index{Green
  relations} in a semigroup (see~\cite{Howie:1976} for example). Let
$S$ be a semigroup and let $S^{1}=S\cup 1$ be the monoid obtained by
adding an identity to $S$.  Two elements $s,t$ of $S$ are
$\R$-equivalent if $sS^{1}=tS^{1}$. They are $\L$-equivalent if
$S^{1}s=S^{1}t$.  It is a classical result (see~\cite{Howie:1976})
that ${\L\R}={\R\L}$ .  Thus ${\L\R}={\R\L}$ is an equivalence on the
semigroup $S$ called the $\D$-equivalence.  A class of the $\R,\L$ or
$\D$-equivalence is called an $\R,\L$ or $\D$-class\index{$\D$-class}.
An \emph{idempotent}\index{idempotent} of $S$ is an element $e$ such
that $e^2=e$.  A $\D$-class is \emph{regular}\index{regular
  $\D$-class}\index{$\D$-class!regular} if it contains an idempotent.
The equivalence $\H$ is defined as $\H=\R\cap\L$. It is classical
result that the $\H$-class of an idempotent is a group. The $\H$-class
of idempotents in the same $\D$-class are isomorphic groups. The
\emph{structure group}%
\index{structure group}\index{$\D$-class!structure group} of a regular
$\D$-class is any of the $\H$-classes of an idempotent of the
$\D$-class.

When $S$ is a semigroup of partial maps from a set $Q$ into itself,
each element of $S$ has a rank which is the cardinality of its image.
The elements of a $\D$-class all have the same rank, which is called
the \emph{rank}\index{rank of a $\D$-class}\index{$\D$-class!rank} of
the $\D$-class.  There is at most one element of rank $0$ which is the
\emph{zero} \index{zero in a semigroup} of the semigroup $S$ and is
denoted $0$.

A \emph{fixpoint}\index{fixpoint} of a partial map $s$ from $Q$ into
itself is an element $q$ such that the image of $q$ by $s$ is $q$. The
rank of an idempotent is equal to the number of its fixpoints. Indeed,
in this case, every element in the image is a fixpoint.

The preorder $\le_\J$ on $S$ is defined by $s\le_\J t$ if
$S^{1}sS^{1}\subset S^{1}tS^{1}$.
Two elements $s,t\in S$ are $\J$-equivalent if
$S^{1}sS^{1}=S^{1}tS^{1}$. One has $\D\subset \J$ and it is a classical
result that in a finite semigroup $\D=\J$. The preorder $\le_\J$ induces
a partial order on the $\D$-classes, still denoted $\le_\J$.

We associate with $\A$ a labeled graph $G(\A)$ called its
\emph{syntactic graph}\index{syntactic graph}.
The vertices of $G(\A)$ are the regular 
$\D$-classes of the transition semigroup of $\A$.  
Each vertex is labeled by the rank of the
$\D$-class and its structure group.  There is an edge from the vertex
associated with a $\D$-class $D$ to the vertex associated to a
$\D$-class $D'$ if and only if $D\ge_{\J}D'$.
\begin{example}
The automaton $\A$ of Figure~\ref{figSyntacticGraph} on the left
is the Fischer automaton of the even shift (Example~\ref{ExMinAutomaton}).
The semigroup of transitions of $\A$ has $3$ regular $\D$-classes of ranks
$2$ (containing $\varphi_\A(b)$), $1$ (containing $\varphi_\A(a)$), and
$0$
(containing $\varphi_\A(aba)$).
Its syntactic graph is
  represented on the right.
\begin{figure}[hbt]
  \centering
\begin{picture}(100,10)(0,-5)
\put(5,0){
\begin{picture}(20,10)
\node(1)(0,0){$1$}\node(2)(20,0){$2$}
\drawloop[loopangle=180](1){$a$}\drawedge[curvedepth=3](1,2){$b$}
\drawedge[curvedepth=3](2,1){$b$}
\end{picture}
}
\put(45,0){
  \begin{picture}(50,10)
    \gasset{Nadjust=wh}
    \node(1)(0,0){$2$, $\Z/2\Z$}\node(2)(25,0){$1$, $\Z/\Z$}
    \node(3)(50,0){$0$, $\Z/\Z$}
    \drawedge(1,2){}\drawedge(2,3){}
  \end{picture}
}
\end{picture}
  \caption{The syntactic graph of the even shift}\label{figSyntacticGraph}
\end{figure}
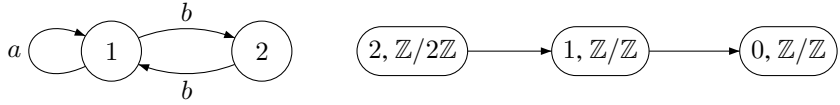
\end{example}
The following result shows that one may reduce to the case
of essential automata.
\begin{proposition}\label{propEssentialPart}
  The syntactic graphs of an automaton and of its essential part are
  isomorphic.
\end{proposition}
\begin{proof}
Let $\A=(Q,E)$ be a deterministic automaton on the alphabet $A$
and let $\A'=(Q',E')$
be its essential part. Let $w\in A^+$ be such that $e=\varphi_\A(w)$
is an idempotent. Then any fixpoint of $e$ is in $Q'$ and thus
$e'=\varphi_{\A'}(w)$ an idempotent of the same rank as $e$.
This shows that $G(\A)$ and $G(\A')$ are isomorphic.
\end{proof}

The following result shows that the syntactic graph characterizes
irreducible shifts of finite type. 

\begin {proposition}\label{propSyntGraphSFT}
  A sofic shift (resp. an irreducible sofic shift)
 is of finite type if and only if the
  syntactic graph of its Krieger automaton (resp. its Fischer automaton)
has nodes of  rank at most $1$.
\end{proposition}
In the proof, we use the following
classical
property of finite semigroups. 
\begin{proposition}\label{propNilSimple}
Let $S$ be a finite semigroup and let $J$ be an ideal of $S$.
 The following conditions are equivalent.
\begin{enumerate}
\item[(i)] All idempotents of $S$ are in $J$.
\item[(ii)] There exists an integer $n\ge 1$ such that $S^n\subset J$.

\end{enumerate}
\end{proposition}
\begin{proof}
Assume that (i) holds.
Let $n=\Card(S)+1$ and let $s=s_1s_2\cdots s_n$ with $s_i\in S$. Then there exist $i,j$ with
$1\le i< j\le n$ such that $s_1s_2\cdots s_i=s_1s_2\cdots s_i\cdots
s_j$.
Let $t,u\in S^1$ be defined by
$t=s_1\cdots s_i$ and $u=s_{i+1}\cdots s_j$. Since $tu=t$,
we have $tu^k=t$ for all $k\ge 1$. Since $S$ is finite, there is a
$k\ge 1$ such that $u^k$ is idempotent and thus $u^k\in J$. This
implies that $t\in J$ and thus $s\in J$. Thus (ii) holds.\\
\noindent It is clear that (ii) implies (i).
\end{proof}
\begin{proof}[Proof of Proposition~\ref{propSyntGraphSFT}]
Let  $X$ be a shift space (resp. an irreducible shift space), let $\A$ be its Krieger
automaton (resp. its Fischer automaton) and let $S$ be the transition semigroup of $\A$. 

If $X$ is of finite type, by
Proposition~\ref{BBEP:propLocalisSFT},
the automaton $\A$ is local.
Any  idempotent in $S$ has
rank $1$ and thus the condition is satisfied. 

Conversely, assume that the graph $G(\A)$ has nodes of rank
at most $1$. Let
 $J$ be the ideal of $S$ formed of the elements of rank at most $1$.
Since all idempotents of $S$ belong to $J$, by Proposition~\ref{propNilSimple},
the semigroup $S$
satisfies
$S^n=J$ for some $n\ge 1$. This shows that for any sufficiently long word
$x$, the map $\varphi_\A(x)$ has rank at most $1$. Thus for
$p,q,r,s\in Q$, if
 $p\cdot x=r$ and $q\cdot x=s$ then $r=s$. This implies that
$\A$ is $(n,0)$-local.
\end{proof}

The following result is from~\cite{Beal&Fiorenzi&Perrin:2006}.

\begin{theorem}\label{theoremSyntGraphAutomata}
 Two symbolic conjugate automata have isomorphic syntactic graphs.
\end{theorem}

We use the following intermediary result.
\begin{proposition}\label{propSyntGraphBipartite}
Let $\A=(\A_1,\A_2)$ be a bipartite automaton.
The syntactic graphs of $\A,\A_1$ and $\A_2$ are isomorphic.
\end{proposition}
\begin{proof}
Let $Q=Q_1\cup Q_2$ and $A=A_1\cup A_2$ be the partitions
of the set of of states and of the alphabet of $\A$ corresponding
to the decomposition $(\A_1,\A_2)$. Set $B_1=A_1A_2$ and $B_2=A_2A_1$.
The semigroups $S_1=\varphi_{\A_1}(B_1^+)$ and $S_2=\varphi_{\A_2}(B_2^+)$ are
included in the semigroup $S=\varphi_\A(A^+)$. Thus the
Green relations of $S$ are refinements of the corresponding Green
relations
in $S_1$ or in $S_2$.
Any idempotent $e$ of $S$ belongs either to $S_1$ or
to $S_2$. Indeed, if $e=0$ then $e$ is in $S_1\cap S_2$. Otherwise,
it has at least one fixpoint $p\in Q_1\cup Q_2$. If $p\in Q_1$, then
$e$ is in $\varphi_A(B_1^+)$ and thus $e\in S_1$. Similarly
if $p\in Q_2$ then $e\in S_2$. 

Let $e$ be an idempotent in $S_1$ and let
$e=\varphi_\A(u)$. Since $u\in B_1^+$, we have $u=au'$
with $a\in A_1$ and $u'\in B_2^*A_2$. Let $v=u'a$.
 Then
$f=\varphi_\A(v)^2$ is idempotent. Indeed, we have
\begin{displaymath}
\varphi_\A(v^3)=\varphi_\A(u'au'au'a)=\varphi_\A(u'uua)=\varphi_\A(u'ua)=\varphi_\A(v^2)
\end{displaymath}
 Moreover
$e,f$ belong the same $\D$-class. Similarly, if $e\in S_2$, there is
an idempotent in $S_1$ which is $\D$ equivalent to $e$. This shows that
a regular $\D$-class of $\varphi_\A(A^+)$ contains idempotents in $S_1$
and in $S_2$.

Finally, two elements of $S_1$ which are $\D$-equivalent in $S$
are also $\D$-equivalent in $S_1$. Indeed, let $s,t\in S_1$ be such
that
$s\R\L t$. Let $u,u',v,v'\in S$ be such that 
\begin{displaymath}
suu'=s,\quad v'vt=t,\quad su=tv
\end{displaymath}
in such a way that $s\R su$ and $vt\L t$.
 Then  $su=vt$ implies that $u,v$ are both
in $S_1$. Similarly $suu'=s$ and $v'vt=t$ imply that $u'v'\in S_1$.
 Thus $s\D t$ in $S_1$. This shows that a regular $\D$ class $D$ of
$S$ contains exactly one $\D$-class $D_1$ of $S_1$ (resp. $D_2$ of $S_2$). Moreover,
an $\H$-class of $D_1$ is also an $\H$-class of $D$.

Thus the three syntactic graphs are isomorphic.
\end{proof}

\begin{proof}[Proof of Theorem~\ref{theoremSyntGraphAutomata}]
  Let $\A=(Q,E)$ and $\B=(R,F)$ be two symbolic conjugate automata on
  the alphabets $A$ and $B$, respectively.  By the Decomposition
  Theorem (Theorem~\ref{DecompTheorAutomata}), we may assume that the
  symbolic conjugacy is a split or a merge.  Assume that $\A'$ is an
  in-split of $\A$. By Proposition~\ref{propEssentialPart}, we may
  assume that $\A$ and $\A'$ are essential.  By
  Proposition~\ref{propInSplitAutomaton}, the adjacency matrices of
  $\A$ and $\A'$ are symbolic elementary equivalent.

  By Proposition~\ref{propElemEquivIsSim}, there is a bipartite
  automaton $\C=(\C_1,\C_2)$ such that $M(\C_1),M(\C_2)$ are similar
  to $M(\A),M(\B)$ respectively. By
  Proposition~\ref{propSyntGraphBipartite}, the syntactic graphs of
  $\C_1,\C_2$ are isomorphic. Since automata with similar adjacency
  matrices have obviously isomorphic syntactic graphs, the result
  follows.
\end{proof}

A refinement of the syntactic graph which is also invariant by
flow equivalence has been introduced in ~\cite{Costa:2006}. The
vertices
of the graph are the \emph{idempotent-bound} $\D$ classes, where
an element $s$ of a semigroup $S$ is called idempotent-bound
if there exist idempotents $e,f\in S$ such that $s=esf$. The
elements of a regular $\D$-class are idempotent-bound.

\intertitre{Flow equivalent automata}
Let $\A$ be an automaton on the alphabet $A$ and let $G$ be its
underlying graph. An \emph{expansion}
\index{expansion!automaton}\index{automaton!expansion} of $\A$
is a pair $(\varphi,\psi)$ of a graph expansion of $G$ and a symbol
expansion of $L_\A$ such that the diagram below is commutative.
\begin{figure}[hbt]
\centering
\gasset{Nframe=n}
\begin{picture}(20,20)
\node(XA)(0,20){$X_\A$}\node(XB)(20,20){$X_\B$}
\drawedge(Xa,Xb){$\varphi$}
\node(La)(0,0){$L_\A$}\node(Lb)(20,0){$L_\B$}
\drawedge(Xa,La){$\lambda_\A$}\drawedge(Xb,Lb){$\lambda_\B$}
\drawedge(La,Lb){$\psi$}
\end{picture}
\end{figure}
The inverse of an automaton expansion is called a contraction.
\begin{example}
Let $\A$ and $\B$ be the automata represented on 
Figure~\ref{figFlowEquivAutomata}. The second automaton is an
expansion
of the first one.
\begin{figure}[hbt]
\centering
\begin{picture}(80,15)(0,-5)
\node(1)(0,5){$1$}
\node(2)(20,5){$2$}
\node(3)(45,5){$3$}
\node(4)(65,0){$4$}\node(5)(65,10){$5$}
\node(6)(85,5){$6$}
\drawedge[curvedepth=3](1,2){$a$}\drawedge[curvedepth=3](2,1){$a$}
\drawloop[loopangle=0](2){$b$}
\drawedge(3,4){$a$}\drawedge(4,6){$\omega$}
\drawedge[linecolor=red,ELside=r](5,3){$\omega$}\drawedge[linecolor=red](4,6){$\omega$}
\drawedge[ELside=r](6,5){$a$}
\drawloop[loopangle=0](6){$b$}
\end{picture}
\caption{An automaton expansion}\label{figFlowEquivAutomata}
\end{figure}
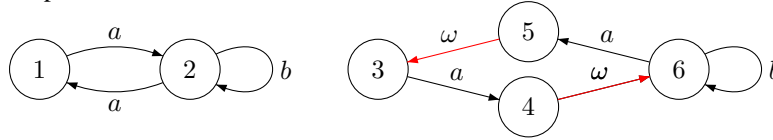
\end{example}

The \emph{flow equivalence} of automata is the equivalence generated
by symbolic conjugacies, expansions and contractions. 

Theorem~\ref{theoremSyntGraphAutomata} has been generalized by
Costa and Steinberg~\cite{CostaSteinberg:2010} to flow equivalence.

\begin{theorem}\label{FlowEquivAuto}
  Two flow equivalent automata have isomorphic syntactic graphs.
\end{theorem}

\begin{example}
  The syntactic graphs of the automata $\A$, $\B$ of
  Example~\ref{ExConjugateSofic} are isomorphic to the syntactic graph
  of the Fischer automaton $\cal C$ of the even shift.  Note that the automata
  $\A, \B$ are not flow equivalent to $\cal C$ . Indeed, the edge
  shifts 
$X_\A$, $X_\B$ on the underlying graphs of the automata
  $\A$, $\B$ are flow equivalent to the full shift on $3$ symbols
  while the edge shift $X_\C$ is flow equivalent to the full shift on $2$
  symbols. Thus the converse of Theorem~\ref{FlowEquivAuto}
  is false.
\end{example}

%
%
\subsection{Pseudovarieties}
In this subsection, we will see how one can formulate characterizations
of some classes of sofic shifts by means of properties of their syntactic
semigroup.
In order to formulate these syntactic characterizations 
of sofic shifts, we introduce the notion of pseudovariety of ordered
semigroups. For a systematic exposition, see the original articles
\cite{Pin:1995},
\cite{PinPinguetWeil:2002}, or the surveys in \cite{Pin:1997} or
\cite{PerrinPin:2004}.

A morphism of ordered semigroups $\varphi$ from $S$ into $T$ is an order
compatible semigroup morphism, that is such that $s\le s'$ implies
$\varphi(s)\le\varphi(s')$.
An ordered subsemigroup of $S$ is a
subsemigroup equipped with the restriction of the preorder.

A \emph{pseudovariety}\index{pseudovariety} of finite ordered semigroups is a
class of ordered semigroups closed under taking ordered subsemigroups,
finite direct products and image under morphisms of ordered
semigroups.

Let $V$ be a pseudovariety of ordered semigroups. We say that
a semigroup $S$ is \emph{locally} in $V$ if all  the submonoids 
of $S$ are in $V$.
The class of these semigroups is a pseudovariety of ordered semigroups.

The following result is due to Costa~\cite{Costa:2007}.

\begin{theorem}
  Let $V$ be a pseudovariety of finite ordered semigroups containing
  the class of commutative ordered monoids such that every element is
  idempotent and greater than the identity.  The class of shifts
  whose syntactic semigroup is locally in $V$ is invariant under
  conjugacy.
\end{theorem}

The following statements give examples of pseudovarieties satisfying
the
above condition.

\begin{proposition}
An irreducible shift space is of finite type if and only if its
syntactic semigroup is locally commutative.
\end{proposition}

 An \emph{inverse
  semigroup}
is a semigroup which can be represented as a semigroup of partial
one-to-one
maps from a finite set $Q$ into itself. The family of inverse
semigroups
does not form a variety (it is not closed under homomorphic image.
However, according to Ash's theorem~\cite{Ash:1987}, the variety
generated by inverse semigroups is characterized by the property
that the idempotents commute.
Using this result, the following result is proved in~\cite{Costa:2007}.
\begin{theorem}
  An irreducible shift space is of almost finite type if and only if
  its syntactic semigroup is locally in the pseudovariety generated by
  inverse semigroups.
\end{theorem}

The fact that  shifts of almost finite type satisfy this condition was proved
in~\cite{Beal&Fiorenzi&Perrin:2006}. The converse was conjectured in
the same paper.

In \cite{CostaSteinberg:2010} it is shown that this result implies
that the class of shifts of almost finite type is invariant under flow
equivalence. This is originally from \cite{FujiwaraOsikawa:1987}.
\bibliographystyle{abbrv}
\addcontentsline{toc}{section}{References}
\begin{footnotesize}
  \bibliography{abbrevs,symbolicdynamics}

\newcommand{\noopsort}[1]{} \newcommand{\singleletter}[1]{#1}
  \newcommand{\etal}{et al.}
\begin{thebibliography}{10}

\bibitem{Ash:1987}
C.~J. Ash.
\newblock Finite semigroups with commuting idempotents.
\newblock {\em J. Austral. Math. Soc. Ser. A}, 43(1):81--90, 1987.

\bibitem{Beal&Fiorenzi&Perrin:2006}
M.-P. B{\'e}al, F.~Fiorenzi, and D.~Perrin.
\newblock The syntactic graph of a sofic shift is invariant under shift
  equivalence.
\newblock {\em Internat. J. Algebra Comput.}, 16(3):443--460, 2006.

\bibitem{Beal&Lombardy&Perrin:2008}
M.-P. B{\'e}al, S.~Lombardy, and D.~Perrin.
\newblock Embeddings of local automata.
\newblock In {\em Information Theory, ISIT 2008}, pages 2351--2355. IEEE, 2008.
\newblock To appear in \emph{Illinois J. Math.}

\bibitem{Berstel&Perrin&Reutenauer:2009}
J.~Berstel, D.~Perrin, and C.~Reutenauer.
\newblock {\em Codes and automata}.
\newblock Cambridge University Press, 2009.

\bibitem{Bowen&Franks:1977}
R.~Bowen and J.~Franks.
\newblock Homology for zero-dimensional nonwandering sets.
\newblock {\em Ann. Math. (2)}, 106(1):73--92, 1977.

\bibitem{Boyle:2002}
M.~Boyle.
\newblock Flow equivalence of shifts of finite type via positive
  factorizations.
\newblock {\em Pacific J. of Math.}, 204:273--317, 2002.

\bibitem{Boyle:2008}
M.~Boyle.
\newblock Open problems in symbolic dynamics.
\newblock In {\em Geometric and probabilistic structures in dynamics}, volume
  469 of {\em Contemp. Math.}, pages 69--118. Amer. Math. Soc., Providence, RI,
  2008.

\bibitem{Boyle&Huang:2003}
M.~Boyle and D.~Huang.
\newblock Poset block equivalence of integral matrices.
\newblock {\em Trans. Amer. Math. Soc.}, 355(10):3861--3886 (electronic), 2003.

\bibitem{Costa:2006}
A.~Costa.
\newblock Conjugacy invariants of subshifts: an approach from profinite
  semigroup theory.
\newblock {\em Int. J. Algebra Comput.}, 16:629--655, 2006.

\bibitem{Costa:2007}
A.~Costa.
\newblock Pseudovarieties defining classes of sofic subshifts closed under
  taking equivalent subshifts.
\newblock {\em J. Pure Applied Alg.}, 209:517--530, 2007.

\bibitem{Costa2007b}
A.~Costa.
\newblock {\em Semigroupos Profinitos e Din{\^a}mica Simb{\'o}lica}.
\newblock PhD thesis, Universidade do Porto, 2007.

\bibitem{CostaSteinberg:2010}
A.~Costa and B.~Steinberg.
\newblock Idempotent splitting categories of idempotents associated to
  subshifts are flow equivalence invariants.
\newblock Technical report, 2010.

\bibitem{Fischer:1975}
R.~Fischer.
\newblock Sofic systems and graphs.
\newblock {\em Monatsh. Math.}, 80:179--186, 1975.

\bibitem{Franks:1984}
J.~Franks.
\newblock Flow equivalence of subshifts of finite type.
\newblock {\em Ergodic Theory Dynam. Systems}, 4(1):53--66, 1984.

\bibitem{FujiwaraOsikawa:1987}
M.~Fujiwara and M.~Osikawa.
\newblock Sofic systems and flow equivalence.
\newblock {\em Math. Rep. Kyushu Univ.}, 16(1):17--27, 1987.

\bibitem{Hamachi&Nasu:1988}
T.~Hamachi and M.~Nasu.
\newblock Topological conjugacy for $1$-block factor maps of subshifts and
  sofic covers.
\newblock In {\em Dynamical Systems}, volume 1342 of {\em Lecture Notes in
  Mathematics}, pages 251--260. Springer-Verlag, 1988.

\bibitem{Howie:1976}
J.~M. Howie.
\newblock {\em An introduction to semigroup theory}.
\newblock Academic Press [Harcourt Brace Jovanovich Publishers], London, 1976.
\newblock L.M.S. Monographs, No. 7.

\bibitem{Krieger:1984}
W.~Krieger.
\newblock On sofic systems. {I}.
\newblock {\em Israel J. Math.}, 48(4):305--330, 1984.

\bibitem{Lind&Marcus:1995}
D.~A. Lind and B.~H. Marcus.
\newblock {\em An introduction to symbolic dynamics and coding}.
\newblock Cambridge University Press, 1995.

\bibitem{Nasu1985}
M.~Nasu.
\newblock An invariant for bounded-to-one factor maps between transitive sofic
  subshifts.
\newblock {\em Ergodic Theory Dynam. Systems}, 5(1):89--105, 1985.

\bibitem{Nasu:1986}
M.~Nasu.
\newblock Topological conjugacy for sofic systems.
\newblock {\em Ergodic Theory Dynam. Systems}, 6(2):265--280, 1986.

\bibitem{Nasu:1988}
M.~Nasu.
\newblock Topological conjugacy for sofic systems and extensions of
  automorphisms of finite subsystems of topological markov shifts.
\newblock In {\em Proceedings of Maryland special year in Dynamics 1986--87},
  volume 1342 of {\em Lecture Notes in Mathematics}, pages 564--607.
  Springer-Verlag, 1988.

\bibitem{ParrySullivan:1975}
B.~Parry and D.~Sullivan.
\newblock A topological invariant of flows on {$1$}-dimensional spaces.
\newblock {\em Topology}, 14(4):297--299, 1975.

\bibitem{PerrinPin:2004}
D.~Perrin and J.-E. Pin.
\newblock {\em Infinite Words}.
\newblock Elsevier, 2004.

\bibitem{Pin:1995}
J.-E. Pin.
\newblock Eilenberg's theorem for positive varieties of languages.
\newblock {\em Izv. Vyssh. Uchebn. Zaved. Mat.}, (1):80--90, 1995.

\bibitem{Pin:1997}
J.-E. Pin.
\newblock Syntactic semigroups.
\newblock In {\em Handbook of formal languages, {V}ol.\ 1}, pages 679--746.
  Springer, Berlin, 1997.

\bibitem{PinPinguetWeil:2002}
J.-E. Pin, A.~Pinguet, and P.~Weil.
\newblock Ordered categories and ordered semigroups.
\newblock {\em Comm. Algebra}, 30(12):5651--5675, 2002.

\end{thebibliography}
\end{footnotesize}


\newpage
\begin{abstract}
  This chapter presents some of the links between automata theory and
symbolic dynamics. The emphasis is on two particular points. The
first one is the interplay between some particular classes of
automata, such as local automata and results on embeddings
of shifts of finite type. The second one is the connection between
syntactic semigroups and the classification of sofic shifts
up to conjugacy.
\end{abstract}


\printindex
\end{document}